\documentclass{article}
\usepackage{macros}

\newtheorem{question}{Question}

\title{Tight Information Complexity of the Coin Problem in the Broadcast Model}
\author{Hadi Kazemi and Varun Jog
\vspace{0.2cm}\\
Department of Pure Mathematics and Mathematical Statistics\\
University of Cambridge}
\date{}
\begin{document}

\maketitle

\begin{abstract}
    We study distributed testing of $\ber(\alpha)$ versus $\ber(\beta)$ in the broadcast, or shared-blackboard, model. For protocols with constant advantage, we characterise up to universal constant factors the information complexity under either hypothesis for every pair $\beta<\alpha$. The characterisation shows that the two information costs can be quite different and identifies three parameter regimes, with optimal protocols based respectively on clean samples, a noisy binary symmetric channel, and an asymmetric $Z$-channel. The lower bounds rely on a novel mixed Hellinger--Jensen--Shannon inequality that may be of independent interest. We also characterise the constant-advantage information complexity of testing arbitrary discrete distributions via an optimisation problem over channels, and show that binary-output channels suffice. We obtain bounds for bounded likelihood-ratio distributions, and give general upper bounds in terms of $\chi^2$ divergence. As applications, we recover the broadcast-model set-disjointness lower bound, and derive stronger lower bounds in the multi-pass streaming setting for some problems considered in prior work.
\end{abstract}

\section{Introduction}\label{sec:introduction}

The \emph{coin problem} is the basic binary hypothesis testing task of distinguishing samples from $\ber(\alpha)$ and $\ber(\beta)$. In the centralised setting, where all samples are available to a single algorithm, the likelihood-ratio test gives the optimal decision rule. In this paper, we study the same testing problem in a distributed setting, where the samples are held by separate agents. Since a single agent typically lacks enough information to identify the bias correctly, the agents must aggregate information about their samples by communicating with each other. There are various models of communication considered in the literature, and in this paper we work in the \emph{broadcast model}, also called the \emph{shared blackboard model}. 

In the broadcast model, there are $n$ agents, agent $i$ observes one sample $X_i$, and the agents communicate by writing messages (bits) on a shared blackboard visible to all. A \emph{blackboard protocol} specifies the order in which agents write messages to the blackboard and how they decide what message to write. The order may be adaptive; i.e., the identity of the next agent who writes on the blackboard may depend on what is already written there. The bit written by an agent may depend on the contents on the blackboard, their own sample, and any available public and private randomness. The resulting transcript is denoted by $\Pi$. The protocol succeeds with constant advantage if the distributions of $\Pi$ under the two hypotheses are sufficiently far apart (usually a fixed constant) in the total variation distance.

In this paper, our focus will be on the \emph{information complexity} of the coin problem. Stated simply, the information complexity of a protocol $\Pi$ is the amount of information revealed by the transcript $\Pi$ about the dataset $X$. As noted before, agents \emph{have to} share some information via the blackboard, since otherwise they have no hope of solving the problem. Under hypothesis 0, the amount of revealed information is $I_\alpha(X; \Pi)$, where the subscript $\alpha$ means $X \sim \ber(\alpha)^{\otimes n}$. Likewise, under hypothesis 1, the revealed information is $I_\beta(X; \Pi)$. The \emph{information complexity of the coin problem under $\ber(\alpha)$} is the infimum of $I_\alpha(X; \Pi)$ over all sample sizes and all protocols that solve the coin problem. Likewise, the \emph{information complexity of the coin problem under $\ber(\beta)$} is the infimum of $I_\beta(X; \Pi)$ over all sample sizes and all protocols that solve the coin problem.

While information complexity is a well-motivated and fundamental quantity to study in itself, it turns out that bounds on information complexity are useful in other information-constrained distributed settings. Consider, for example, the communication complexity of distributed algorithms. Let $W$ be the public coin, let $M$ be the sequence of messages written on the blackboard, and let $\Pi=(W,M)$ be the full transcript. If the protocol communicates at most $C$ bits, then
\begin{align*}
    I(X;\Pi)
    =
    I(X;M\mid W)
    \leq
    H(M\mid W)
    \leq
    C.
\end{align*}
Thus, information-complexity lower bounds imply communication lower bounds. Similar ideas also work in the context of distributed differentially private algorithms. Furthermore, as the broadcast model is the most adaptive communication model, lower bounds on the information complexity under the broadcast model also imply lower bounds under less adaptive models of communication, such as sequentially-adaptive or non-adaptive models. The main question we study in this paper is the following:

\begin{question}
\label{ques:1}
What is the information complexity of the coin problem in the broadcast model under each hypothesis?
\end{question}

This question received recent attention in the work of \cite{streaming2025}. That work proved tight bounds, up to universal constant factors, on the information complexity of $\ber(0)$ versus $\ber(1/2)$. They were also able to prove an $\Omega(\min\{1/2, 1-\alpha\})$ lower bound on the information complexity of $\ber(0)$ versus $\ber(\alpha)$ under the $\ber(\alpha)$ distribution.\footnote{The stated bound in \cite{streaming2025} has the weaker form $\Omega(\min\{\alpha, \bar \alpha\})$, but this can be improved to $\Omega(\min\{1/2, 1-\alpha\})$ using the same proof presented in the paper. See \Cref{subsec:coin-comparison} for details.} For small $\alpha$, this bound is $\Omega(1)$. However, intuitively, as $\alpha \to 0$, the coin problem becomes harder and harder. Hence, it seems natural that any blackboard protocol solving this problem must share more information, and therefore the $\Omega(1)$ lower bound ought to be loose. This observation motivated our search for a stronger lower bound in this special case, and led to the broader question posed in \Cref{ques:1}.

A main message of this paper is that the answer depends on where the two biases $\alpha$ and $\beta$ lie in $[0,1]$: in different regimes, the best protocol reveals clean samples, extremely noisy samples, or samples passed through an asymmetric channel. Thus, even for Bernoulli testing, characterising the tight information complexity is non-trivial and exhibits interesting phenomena. 

\subsection{Our contributions}

We summarise our main contributions below.

\paragraph{(i) Exact information complexity of the coin problem:}
We solve \Cref{ques:1} completely in \Cref{thm:info_comp_ber}, identifying the tight information complexity under both hypotheses for constant advantage. For the coin problem of testing $\ber(\alpha)$ versus $\ber(\beta)$, let $\icht_\alpha$ and $\icht_\beta$ denote the minimum information costs, measured under $\ber(\alpha)$ and $\ber(\beta)$ respectively, among all blackboard protocols with constant advantage. Assume without loss of generality that $\beta<\alpha$. The answer depends on where $\alpha$ and $\beta$ lie in $[0,1]$: near $0$, straddling $1/2$, or near $1$. For ease of exposition, we state the expressions for $\icht_\alpha$. Let $\ent(x):=-x\log x-\bar x\log \bar x$ be the binary entropy function and let $\hel(x,y)$ denote the Hellinger divergence between $\ber(x)$ and $\ber(y)$.

\emph{Regime 1 ($0 \le \beta < \alpha < 1/2$):} In this regime, the optimal $\icht_\alpha$ is given by
\begin{align*}
\icht_\alpha \asymp \frac{\ent(\alpha)}{\hel(\alpha, \beta)}.
\end{align*}
The matching protocol is the obvious one: write clean samples on the blackboard. Since $n^*\asymp 1/\hel(\alpha,\beta)$ samples suffice for centralised testing, broadcasting $X_i$ for $i\leq n^*$ gives information cost $I(X;\Pi)=H(\Pi)=n^*\ent(\alpha)$. Our lower bound shows that this clean-bits protocol is optimal up to constants.

\emph{Regime 2 ($\beta < 1/2 < \alpha$):} In this regime, the optimal $\icht_\alpha$ is given by
\begin{align*}
\icht_\alpha \asymp \frac{\bar \alpha}{\hel(\alpha, \beta)}.
\end{align*}
Here writing clean bits is not optimal, since $\ent(\alpha)\gg \bar\alpha$ when $\alpha$ is close to 1. The better protocol writes very noisy versions $Y_i$ of the samples, obtained by passing each $X_i$ through a binary symmetric channel with crossover probability $1/2-\eta$. This changes the Bernoulli parameters to $\alpha_\eta=1/2+\eta(2\alpha-1)$ and $\beta_\eta=1/2-\eta(1-2\beta)$. Although this increases the required sample size to $n^*\asymp 1/\hel(\alpha_\eta,\beta_\eta)$, each message reveals much less information. The total information cost becomes
\begin{align*}
I(X; \Pi) = \sum_{i=1}^{n^*} I(X_i; Y_i) \asymp \frac{\ent(\alpha_\eta)-\ent(1/2+\eta)}{\hel(\alpha_\eta, \beta_\eta)}.
\end{align*}
Letting $\eta\to0$ gives the claimed $\bar\alpha/\hel(\alpha,\beta)$ scaling, and our lower bound shows that this very-noisy-bits protocol is optimal up to constants.

\emph{Regime 3 ($1/2 \le \beta < \alpha$):} In this final regime, the optimal information complexity is given by
\begin{align*}
\icht_\alpha \asymp \frac{\bar \alpha}{\bar \beta} \frac{\ent(\beta)}{\hel(\alpha, \beta)}.
\end{align*}
In this regime, clean bits again lose: they give information complexity $\asymp \ent(\alpha)/\hel(\alpha,\beta)$. By concavity of entropy and $0<\bar\alpha<\bar\beta$,
\begin{align*}
    \ent(\alpha) = \ent(\bar \alpha) \geq \frac{\bar \alpha}{\bar \beta} \ent(\bar \beta) = \frac{\bar \alpha}{\bar \beta} \ent(\beta).
\end{align*}
This can be much larger than the target expression when $\alpha$ is close to $1$. The very-noisy-bits strategy from regime 2 is also suboptimal, giving $\asymp (\bar\alpha/\bar\beta)/\hel(\alpha,\beta)$, which can be too large when $\beta\to1$. The optimal protocol is intuitive: In this regime, 0s are rare and hence reveal a lot of information, but 1s are plentiful and are not as informative. A natural strategy is to write the 1s cleanly and 0s noisily, for instance by writing an independent $\ber(\gamma)$-sample corresponding to each 0 for some $\gamma \in [0,1]$. Such a channel is called a $Z$-channel in information theory. The $Z$-channel obtained by setting $\gamma = \beta$ turns out to be the correct choice. We defer the detailed calculation to the proof of \Cref{thm:info_comp_ber}.

Thus, even achievability already exhibits the main phenomenon of different strategies being optimal in different regimes. The harder part is proving that no interactive blackboard protocol can do better. For this we revisit the Hellinger cut-and-paste framework of \cite{streaming2025,jayram2009hellinger}. Our main technical contribution is tightening this analysis by replacing a weak inequality therein by a strong one. Specifically, we prove a new lower bound for a skewed Jensen--Shannon divergence in terms of a mixed Hellinger divergence, established in \Cref{thm:mixed-hellinger-js}. This $f$-divergence inequality is the technical core of the paper and may be of independent interest.

\paragraph{(ii) Information complexity of testing $p$ vs $q$:}
We also consider the general hypothesis testing problem of testing $p$ versus $q$ with common domain $[k] = \{1,2,\dots, k\}$. Let $\cC(\cA, \cB)$ be the set of all channels with input alphabet $\cA$ and output alphabet $\cB$. If $X \sim p$, let $Y \sim Cp$ be the output of the channel when the input is $X$. We observe that the tight information complexity (for constant advantage) under $p$ satisfies
\begin{align*}
\icht_p \asymp \inf_{C \in \cC([k], \mathbb N)} \frac{I(X;Y)}{\hel(Cp, Cq)}.
\end{align*}
This observation follows in a straightforward manner from the known lower-bounding techniques established in \cite{jayram2009hellinger}, but we did not find it explicitly written in prior work.

It is unlikely that the optimisation problem on the right-hand side has an analytical expression in general, but we show lower and upper bounds on this value. The lower bound is shown using the data-processing inequality and the analytical expressions for the information complexity of the coin problem derived in (i). The upper bound is shown by calculating the information complexity when using a ``very noisy channel'' in $\cC([k],[2])$ that is suitably tuned using the knowledge of $p$ and $q$. Finally, we show that to compute $\icht_p$, it is enough to consider the smaller class of channels $\cC([k], [2])$ instead of $\cC([k], \mathbb N)$, making the resulting optimisation problem finite dimensional for any given $p$ and $q$.

\paragraph{(iii) Applications:} Information-complexity tools were recently used to derive memory lower bounds for streaming algorithms. Brown, Bun, and Smith \cite{brown2022strong} derived such lower bounds using a core task, which is a hypothesis testing problem between a structured distribution and an unstructured distribution. In their construction, the background is $\ber(1/2)^{\otimes d}$, and each hidden coordinate is fixed across the stream to an independent $\ber(1/2)$ value. Motivated by applications to learning from sparse features, we consider the case where the background is $\ber(\nu)^{\otimes d}$ and each hidden value is drawn once from $\ber(\nu)$. When $\nu \ll 1$, such as $\log d/d$, the data have roughly $\log d$ non-zero entries. We show that the information-complexity lower bound for this version gains an extra $\log(1/\nu)$ factor. This lower bound may be used for further downstream applications as in \cite{brown2022strong}, although we do not pursue this in this paper.

A second application is to the communication complexity of the $k$-party set-disjointness problem studied in \cite{braverman2015information}, via information-complexity direct-sum arguments. Here, each of $k$ agents holds a set $X_i \subseteq \{1,2,\dots,n\}$, and they wish to determine whether $\bigcap_{i=1}^k X_i = \emptyset$ in the broadcast model. Equivalently, representing each set by a bit string in $\{0,1\}^n$, the agents wish to compute
\begin{align*}
\bigvee_{j \in [n]} \bigwedge_{i \in [k]} X_{ij}.
\end{align*}
The key technical step in \cite{braverman2015information} is a conditional information-complexity lower bound for computing a single $\mathrm{AND}$ function under the following hard distribution: pick $Z \in [k]$ uniformly at random, set $X_Z = 0$, and draw the remaining players' bits independently from $\ber(1-1/k)$. Through a careful analysis they show that any protocol computing $\mathrm{AND}$ with constant advantage must have $I(X;\Pi\mid Z)\gtrsim \log k$. Our results give a significantly simpler route to the same using the tight $\Theta(\log k)$ information complexity, measured under the $\ber(1/\sqrt k)$ distribution, for the coin problem of testing $\ber(0)$ versus $\ber(1/\sqrt{k})$. Combined with known direct-sum arguments, this leads to a $\Omega(n\log k)$ lower bound. 

\subsection{Related work}
As noted earlier, a special case of the main question in our paper was studied in~\cite{streaming2025}, although the main focus of that work was to study statistical inference in streaming models by using the coin problem (called the bit-bias problem in their work) as a primitive. The coin problem primitive for streaming settings also appeared in an earlier work~\cite{brown2022strong}, which established fundamental memory lower bounds for several natural learning problems using the basic building block of testing $\ber(0)$ versus $\ber(1/2)$. \cite{brown2022strong} relied on a notion of information complexity in the streaming model proposed in~\cite{braverman2020coin} (see also \cite{braverman2024newinfocomp}). A key insight of~\cite{streaming2025} was the observation that it is enough to prove lower bounds for the information complexity in the broadcast model, as these bounds continue to hold for the information complexity expression in the streaming model. Using this, \cite{streaming2025} proved tight bounds on the information complexity of testing $\ber(0)$ versus $\ber(1/2)$ in the broadcast model. The strategy of proving broadcast model lower bounds to establish streaming model lower bounds for different statistical problems also appears in \cite{dagan2018detecting,garg2026unified}. We also note that the coin problem in the streaming setting has been studied without recourse to information complexity in prior work~\cite{Cov69, Kop75, steinberger2013distinguishability, brody2010coin, braverman2022tight}. We refer the reader to~\cite{braverman2022tight} for a more complete discussion of the coin problem along these lines.

We next place this closest line of work in two broader contexts: information-complexity lower bounds in communication complexity, and distributed hypothesis testing under communication or information constraints.

More broadly, our focus on the information complexity of the coin problem---a statistical inference problem---has been studied in prior work mainly as a tool to analyse the communication complexity of inference problems~\cite{woodruff2012tight,bar2004information,braverman2013tight,braverman2015information,hadar2019communication}. Another reason for studying information complexity is that it is indispensable when proving direct-sum theorems, which are used to show lower bounds on the complexity of simultaneously solving multiple copies of a problem~\cite{bravermanICM}. This approach has been used to give sharp lower bounds on the communication complexity of the set-disjointness problem, which is an important problem in communication complexity theory, in both the two-party \cite{bar2004information,jayram2009hellinger} and broadcast models \cite{braverman2015information}.

A second related area is distributed simple binary hypothesis testing, i.e., testing $p$ versus $q$ when observations are spread across multiple agents. A classical formulation of this problem uses a fusion center, a central server that receives messages from the agents~\cite{Tsi93}. This problem has been studied extensively in recent years under information constraints such as privacy and communication. Analytical and algorithmic results were presented in~\cite{KOV16, PenEtal23, PenEtal24, PenEtal24b, KazEtal25, GNCI26}. Non-adaptive and sequentially adaptive protocols were shown to be roughly equally powerful in this setting~\cite{KazEtal25}, but whether blackboard protocols confer any advantages is as yet unclear. However, these results are only non-trivial in settings where $p$ and $q$ are supported on $\{1, 2, \dots, k\}$ for $k$ much larger than 2.

Finally, the coin problem also has a rich history in theoretical computer science owing to its connection to computing the majority and threshold functions. Early work appears in \cite{shaltiel2010hardness,aaronson2010bqp}, and the term ``coin problem'' was coined in a FOCS paper~\cite{brody2010coin}. Any computational model that can compute the majority of $n$ bits can also solve the coin problem of testing $\ber(1/2 + \eps)$ versus $\ber(1/2 - \eps)$ up to the information-theoretic threshold of $\eps \asymp 1/\sqrt n$. In resource-constrained models, however, it may be impossible to compute the majority function exactly or even approximately. In such models, complexity theory explores more refined questions such as identifying the smallest distinguishable bias $\eps$ or the smallest amount of resources, such as sample size $n$, circuit size, branching program width, or memory necessary to solve the problem at a fixed $\eps$. Several prior works have considered these questions in models such as $\mathsf{AC}^0$ circuits~\cite{CohEtal14}, $\mathsf{AC}^0[\oplus]$ circuits~\cite{LimEtal21}, product tests~\cite{LeeVio18}, and read-once branching programs (ROBPs)~\cite{brody2010coin, steinberger2013distinguishability, braverman2022tight}.

\paragraph{Notation:} For a positive integer $n$, write $[n]=\{1,\dots,n\}$. All distributions in the paper are supported on discrete spaces. We use $\cX$ for a generic input space, and $\cP,\cQ,\cR$ for distributions on that space. If $X_1,\dots,X_n$ are independent samples from $\cP$, we write $X\sim\cP^{\otimes n}$, and use $\ber(\alpha)$ for the Bernoulli distribution with mean $\alpha$. For $x\in[0,1]$, write $\bar x=1-x$; in particular, $\bar\alpha=1-\alpha$ and $\bar\beta=1-\beta$. We write $\ent(\cdot)$ for Shannon entropy, $\mi(\cdot;\cdot)$ for mutual information, $\dkl(\cdot\|\cdot)$ for Kullback--Leibler divergence, $\dtv(\cdot,\cdot)$ for total variation distance, $\hel(\cdot,\cdot)$ for Hellinger divergence, and $\chisq(\cdot,\cdot)$ for chi-squared divergence. The set of channels from an input alphabet $\cA$ to an output alphabet $\cB$ is denoted by $\cC(\cA,\cB)$; applying a channel $C$ to a distribution $\cP$ gives the output distribution $C\cP$. We use standard asymptotic notation, and write $f\lesssim g$, $f\gtrsim  g$, and $f\asymp  g$ when the implicit constants are universal. We use $\cC([k],\nats)$ as shorthand for $\bigcup_{m\in\nats}\cC([k],[m])$: every channel in this class has a finite output alphabet, but there is no uniform bound on its size.

\paragraph{Paper structure:} The rest of the paper is organised as follows. \Cref{sec:preliminaries} collects the background on binary hypothesis testing, divergences, blackboard protocols, information complexity, and the Hellinger cut-and-paste framework used in the lower bounds. \Cref{sec:coin-information} proves the main result of the paper: a tight characterisation of the information complexity of the coin problem under both hypotheses, including the three parameter regimes and the matching protocols. \Cref{sec:general-information} discusses the information complexity for testing between arbitrary distributions. \Cref{sec:applications} gives the applications to streaming settings and to broadcast-model set-disjointness. \Cref{sec:conclusion} concludes with open questions and directions for future work.

\section{Preliminaries}\label{sec:preliminaries}
This section collects preliminary results needed for the rest of the paper. We first define binary hypothesis testing and related $f$-divergences. We then define blackboard protocols and information complexity, and finally review the Hellinger-based lower-bound framework that underlies our proofs.

\subsection{Divergences in simple binary hypothesis testing}

The coin problem can be viewed as a special case of simple binary hypothesis testing, which we define as follows.

\begin{definition}\label{def:simple_bht}
Consider two distributions $\cP$ and $\cQ$ supported on a discrete set $\cX$. Let $\Theta$ be uniformly distributed on the two-point set $\{\cP,\cQ\}$. Conditioned on $\Theta=\cR$, the random vector $X=(X_1,\dots,X_n)$ is distributed as $\cR^{\otimes n}$. Let $\Pi$ be a possibly randomised map from $\cX^n$ to an arbitrary discrete set $\boldsymbol{\Pi}$. Let $\Pi_\cP$ denote the distribution of $\Pi(X)$ when $X\sim \cP^{\otimes n}$, and define $\Pi_\cQ$ analogously. We say that $\Pi$ solves the hypothesis testing problem with advantage $\delta\in(0,1)$ if $\dtv(\Pi_\cP,\Pi_\cQ)\geq \delta$.
\end{definition}

Observe that $\Pi$ can be used to test between the two hypotheses with error $\frac{1-\dtv(\Pi_\cP,\Pi_\cQ)}{2}$. Usually, the output space $\boldsymbol{\Pi}$ is simply $\{\cP,\cQ\}$, but in this paper we often view $\Pi$ as a protocol transcript and hence allow its range to be an arbitrary discrete set. We are often interested in constant advantage, such as $\delta=1/3$. In simple binary hypothesis testing, many fundamental quantities, including the optimal Bayes error, sample complexity, and asymptotic error rates, are characterised by appropriate $f$-divergences. An $f$-divergence is defined as follows.

\begin{definition}[$f$-divergence \cite{powu2025}]\label{def:f-div}
Let $f:[0,\infty)\to\real\cup\{\infty\}$ be a convex function with $f(1)=0$, and let $f'(\infty):=\lim_{t\to\infty} f(t)/t$. For two distributions $\cP$ and $\cQ$, we define
\begin{equation*}
D_f(\cP\parallel \cQ)
=
\E\left[
f\left(\frac{d\cP}{d\cQ}(Z)\right)
\mathbbm{1}_{{\frac{d\cP}{d\cQ}(Z)<\infty}}
\right]
+
f'(\infty)\cP\left(\left\{x:\frac{d\cP}{d\cQ}(x)=\infty\right\}\right),
\end{equation*}
where $Z\sim\cQ$.
\end{definition}

Of particular importance to this paper is the Hellinger divergence.

\begin{definition}[Hellinger divergence]\label{def:hellinger}
For two distributions $\cP$ and $\cQ$ over a discrete domain $\cX$, the Hellinger divergence is defined as $\hel(\cP,\cQ)\coloneqq \frac{1}{2}\sum_{x\in\cX}\left(\sqrt{\cP(x)}-\sqrt{\cQ(x)}\right)^2$. Equivalently, $\hel(\cP,\cQ)=1-\sum_{x\in\cX}\sqrt{\cP(x)\cQ(x)}$.
\end{definition}

\begin{definition}[Skewed Jensen--Shannon divergence]\label{def:skewed-js}
For $\alpha\in[0,1]$ and two distributions $\cP$ and $\cQ$ on the same discrete space, the $\alpha$-skewed Jensen--Shannon divergence is
\begin{equation*}
\js_\alpha(\cP,\cQ)
\coloneqq
\alpha\dkl(\cP\|\alpha\cP+\alb\cQ)
+
\alb\dkl(\cQ\|\alpha\cP+\alb\cQ).
\end{equation*}
\end{definition}

Equivalently, $\js_\alpha(\cP,\cQ)$ is the mutual information $\mi(U;Y)$ when $U\sim\ber(\alpha)$ and, conditionally on $U$, the output $Y$ has law $\cP$ if $U=1$ and law $\cQ$ if $U=0$.

The following bound between $\dtv$ and $\hel$ is standard.

\begin{fact}\label{fact:tv-hel-ineq}
For any two distributions $\cP$ and $\cQ$, we have
\begin{equation*}
\hel(\cP,\cQ)
\geq
\frac{1}{2}\dtv^2(\cP,\cQ).
\end{equation*}
\end{fact}

The Hellinger divergence between Bernoulli distributions satisfies the following useful estimate.

\begin{fact}[\cite{PenEtal24b}]\label{fact:hel_ber}
For any $\alpha,\beta\in[0,1]$ such that $\min(\alpha,\beta)\leq 1/2$, we have
\begin{equation*}
\hel(\ber(\alpha),\ber(\beta))
\asymp
\frac{(\alpha-\beta)^2}{\max\{\alpha,\beta\}}.
\end{equation*}
\end{fact}

In this paper, we occasionally require the sample complexity of simple binary hypothesis testing. We define it and recall its standard characterisation below.

\begin{definition}[Sample complexity]\label{def:hyp_sample_comp}
Consider the simple binary hypothesis testing problem from \Cref{def:simple_bht}. For $\delta\in(0,1)$, the sample complexity $n^*(\cP,\cQ,\delta)$ is the smallest number of i.i.d.\ samples for which the optimal likelihood-ratio test solves the testing problem with advantage at least $\delta$. When the arguments are clear from context, we write $n^*$ or $n^*(\delta)$.
\end{definition}

\begin{fact}[\cite{LeCam86} (see also \cite{PenEtal24b})]\label{fact:hyp_sample_comp_fact}
Suppose that $\hel(\cP,\cQ)\leq 1/2$. Then, for $\delta \in [1/3,1)$, the sample complexity of simple binary hypothesis testing between $\cP$ and $\cQ$ satisfies $n^*(\delta)\asymp \frac{\log \frac{1}{1-\delta}}{\hel(\cP,\cQ)}$. For $\delta \in (0, 1/3)$, the sample complexity satisfies the bounds $\frac{\delta^2}{\hel(\cP, \cQ)} \lesssim n^*(\delta) \lesssim \frac{\delta}{\hel(\cP, \cQ)}$. 
\end{fact}
For small advantage, the sample complexity depends on factors other than just the Hellinger divergence; see \cite[Example 8.1]{PenEtal24b}. In this paper, we are primarily interested in constant advantage, but we state our bounds for general $\delta$.

Finally, we note the following variational form for $\chi^2$ divergence.

\begin{fact}\label{fact:chisq_var}
For any two distributions $\cP$ and $\cQ$ with $\cP\ll \cQ$, we have
\begin{equation*}
\chisq(\cP,\cQ)
=
\sup_{\substack{f:\cX\to[-1,1]\ \var_\cQ(f(X))>0}}
\frac{\left(\E_\cP f(X)-\E_\cQ f(X)\right)^2}{\var_\cQ(f(X))}.
\end{equation*}
\end{fact}

\subsection{Problem setup}\label{sec:problem_setup}

We now describe the distributed hypothesis testing framework and the blackboard communication protocol. Recall that there are $n$ agents, where agent $i$ observes a sample $X_i\in\cX$ for each $i\in[n]$. 

\begin{definition}[Blackboard communication protocol]\label{def:blackboard_protocol}
    A public-coin blackboard communication protocol $\Pi$ is one in which all agents observe a random variable $W$ that is independent of $X$. Conditioned on each realisation $W=w$, the protocol is specified by a finite rooted binary tree with vertex set $\cV(\Pi)$. Each non-leaf vertex $v\in\cV(\Pi)$ is labelled by an agent $\sigma(v)\in[n]$ and is associated with a binary channel $K_v:\cX\to\Delta(\{0,1\})$. Each edge is labelled either $0$ or $1$, and every non-leaf vertex has exactly two children, one connected by a $0$-labelled edge and the other by a $1$-labelled edge.

    The protocol starts at the root. At a non-leaf vertex $v\in\cV(\Pi)$, agent $\sigma(v)$ broadcasts a bit $B_v$ sampled from $K_v(\cdot\mid X_{\sigma(v)})$, and the protocol moves to the child along the edge labelled $B_v$. It ends upon reaching a leaf. Let $M(X)$ denote the leaf reached, equivalently the full sequence of messages written on the blackboard. We define the full transcript to be
    \begin{align*}
        \Pi(X)
        =
        (W,M(X)).
    \end{align*}
    As usual, we use $\Pi$ for both the protocol and its random transcript when the input is clear. A private-coin protocol is the special case in which $W$ is constant.
\end{definition}

We now define the information cost of a protocol when the input $X=(X_1,\dots,X_n)$ is drawn i.i.d.\ from some distribution $\cP$.

\begin{definition}[Information cost]\label{def:information_cost}
    The information cost of a public-coin blackboard protocol $\Pi=(W,M)$ under $\cP$ is
    \begin{align*}
        I_\cP(X;\Pi)
        \coloneqq
        I(X;W,M)
        =
        I(X;M\mid W),
    \end{align*}
    where the equality uses the independence of $W$ and $X$.
\end{definition}

For a distribution $\cR$ on $\cX$, write $\Pi_\cR$ for the law of the full transcript when $X\sim\cR^{\otimes n}$. A protocol solves the simple binary hypothesis testing problem between $\cP$ and $\cQ$ with advantage $\delta$ if
\begin{align*}
    \dtv(\Pi_\cP,\Pi_\cQ)
    \geq
    \delta.
\end{align*}
Equivalently, if $M^w_\cR$ denotes the conditional law of the message transcript given $W=w$ and $X\sim\cR^{\otimes n}$, then
\begin{align*}
    \dtv(\Pi_\cP,\Pi_\cQ)
    =
    \E_W\left[\dtv(M^W_\cP,M^W_\cQ)\right].
\end{align*}

\begin{definition}[Information complexity of hypothesis testing]\label{def:info_comp_def}
    Let $\cP$ and $\cQ$ be distributions on $\cX$, let $\delta\in(0,1)$, and let $\cR\in\{\cP,\cQ\}$. The information complexity of testing $\cP$ versus $\cQ$ with advantage $\delta$, measured under $\cR$, is
    \begin{align*}
        \icht_\cR(\cP,\cQ,\delta)
        \coloneqq
        \inf I_\cR(X;\Pi).
    \end{align*}
    Here $X=(X_1,\dots,X_n)\sim\cR^{\otimes n}$. The infimum is over all $n\in\mathbb N$ and all public-coin blackboard protocols on $n$ samples whose full transcript satisfies $\dtv(\Pi_\cP,\Pi_\cQ)\geq\delta$. When the arguments are clear from context, we write $\icht_\cR$ or $\icht_\cR(\delta)$.
\end{definition}

This definition optimises over the number of samples as well as the protocol. Equivalently, it agrees with the limit of the optimal $n$-sample information cost as $n\to\infty$, since a protocol using $n$ samples can be viewed as a protocol using any larger number of samples by ignoring the additional samples.

\subsection{Lower-bound framework}\label{subsection:lower-bound-framework}

Lower bounds on information complexity have appeared in the prior work, and there is a well-established framework to derive such bounds. This framework was introduced in \cite{bar2004information,jayram2009hellinger} and used in \cite{garg2026unified, streaming2025, dagan2018detecting} and several others. In the following lemma, we use the notation $\Pi_\cP(X\subt{i}\cQ)$ for the transcript of protocol $\Pi$ when $X_j\simiid\cP$ for all $j\neq i$, and $X_i\sim\cQ$. The following lemma follows from Jayram's Hellinger cut-and-paste argument \cite[Theorem~7]{jayram2009hellinger}; a closely related product-distribution variant is stated explicitly in \cite[Lemma~22]{streaming2025}.

\begin{lemma}[\cite{jayram2009hellinger}; see also \cite{streaming2025}]\label{lem:cut-paste}
    Let $\Pi$ be an $n$-party protocol and $\cP$ and $\cQ$ be two distributions on $\cX$. We have the following inequality
    \begin{equation*}
        \sum_{i\in[n]}\hel\lp\Pi_\cP(X),\Pi_\cP( X\subt{i}\cQ)\rp
        \geq
        \kappa.\hel\lp\Pi_\cP(X),\Pi_\cQ(X)\rp
    \end{equation*}
    where $\kappa=\prod_{\ell=1}^\infty(1-\frac{1}{2^\ell})\approx 0.288788...$ is known as the digital search tree constant.
\end{lemma}

Combining \Cref{lem:cut-paste} and \Cref{fact:tv-hel-ineq}, any protocol with advantage at least $\delta\in(0,1)$ in testing between $\cP$ and $\cQ$ satisfies
\begin{equation}\label{ineq:sum-hel-lb}
    \sum_{i\in[n]}\hel\lp\Pi_\cP(X),\Pi_\cP( X\subt{i}\cQ)\rp
    \geq
    \frac{\kappa}{2}\delta^2.
\end{equation}
Separately, by independence of the samples and subadditivity of conditional entropy, the information cost satisfies
\begin{equation}\label{ineq:mi-sup-add}
    \mi_\cP(X;\Pi) \geq \sum_{i\in[n]}\mi_\cP(X_i;\Pi).
\end{equation}
Thus, if for every public-coin protocol $\Pi$ and every coordinate $i$, $\mi_\cP(X_i; \Pi)$ satisfies
\begin{align}\label{eq:eta_h}
    \mi_\cP(X_i; \Pi) \gtrsim \eta \hel\lp\Pi_\cP(X),\Pi_\cP( X\subt{i}\cQ)\rp
\end{align}
for some $\eta$ (that may depend on $\cP$ and $\cQ$), then we may conclude the lower bound
\begin{align*}
    \icht_\cP(\delta) \gtrsim \eta \delta^2.
\end{align*}
This is precisely the approach used in prior work. In prior work, the key inequality in~\cref{eq:eta_h} is established using straightforward convexity arguments \cite{bar2004information,jayram2009hellinger,streaming2025}. For example, \cite{streaming2025} prove that when $\cP = \ber(1/2)$ and $\cQ = \ber(0)$, $\eta$ may be taken to be a constant. 

Our key observation is that identifying the largest possible $\eta$ will yield the best possible lower bound on the information cost, and finding this $\eta$ for the coin problem is the main problem tackled in this paper. In~\Cref{sec:coin-information} we shall see that the convexity-based bounds from prior work can be substantially tightened for the coin problem. We remark that a closely-related approach to proving lower bounds avoids computing the best $\eta$ and instead relies on computing the strong post-data processing constant of a particular channel \cite{braverman2016communication}. However, in~\Cref{sec:general-information} we shall show that finding the best $\eta$ is the correct approach, in that it yields not just a lower bound but a tight characterisation in the constant advantage regime. In particular, we observe in \Cref{subsec:coin-comparison} that the SDPI-based lower bounds are loose for the coin problem.

\section{Information complexity of the coin problem}\label{sec:coin-information}

This section proves the main result for the coin problem. We begin by stating the theorem and explaining why it is enough to analyse the information cost under one of the two hypotheses. We then prove the key mixed Hellinger--Jensen--Shannon inequality for binary-input channels. With this result in hand, we prove the lower bounds using the cut-and-paste framework from \Cref{sec:preliminaries}, and then give matching protocols for the three parameter regimes. The section ends by comparing the resulting bounds with the post-SDPI method and with earlier streaming lower bounds.

\subsection{Main theorem and proof strategy}\label{subsec:coin-main-theorem-strategy}

In this section we determine the information complexity of testing $\ber(\alpha)$ versus $\ber(\beta)$ in the broadcast model. Throughout the section we assume, without loss of generality, that $\beta<\alpha$, and we write $\alpha$ and $\beta$ for the corresponding Bernoulli distributions when this causes no confusion. For $\cR\in\{\ber(\alpha),\ber(\beta)\}$, the quantity $\icht_\cR(\alpha,\beta,\delta)$ denotes the minimum information cost, measured under $\cR$, among all public-coin blackboard protocols that distinguish the two hypotheses with advantage at least $\delta$.

The main result gives matching lower and upper bounds in the constant-advantage regime. For arbitrary $\delta$, we keep the dependence on $\delta$ explicit: the lower bounds scale as $\delta^2$, while the upper bounds inherit the dependence on $\delta$ from the centralised sample complexity in \Cref{fact:hyp_sample_comp_fact}.

\begin{theorem}[Information complexity of the coin problem]\label{thm:coin-info-section3}\label{thm:info_comp_ber}
Let $\beta<\alpha\in[0,1]$. For every $\delta\in(0,1)$,
\begin{align*}
    \icht_\alpha(\alpha,\beta,\delta)
    \gtrsim
    \delta^2
    \begin{cases}
    \frac{\ent(\alpha)}{\hel(\alpha,\beta)},&
    \beta<\alpha\leq1/2,
    \vspace{0.1cm}\\
    \frac{\alb}{\hel(\alpha,\beta)},&
    \beta<1/2<\alpha,
    \vspace{0.1cm}\\
    \frac{\alb}{\betb}\frac{\ent(\beta)}{\hel(\alpha,\beta)},&
    1/2\leq\beta<\alpha,
    \end{cases}
    \quad
    \icht_\beta(\alpha,\beta,\delta)
    \gtrsim
    \delta^2
    \begin{cases}
    \frac{\beta}{\alpha}\frac{\ent(\alpha)}{\hel(\alpha,\beta)},&
    \beta<\alpha\leq1/2,
    \vspace{0.1cm}\\
    \frac{\beta}{\hel(\alpha,\beta)},&
    \beta<1/2<\alpha,
    \vspace{0.1cm}\\
    \frac{\ent(\beta)}{\hel(\alb,\betb)},&
    1/2\leq\beta<\alpha .
    \end{cases}
\end{align*}
Conversely, there are explicit blackboard protocols with
\begin{align*}
    \icht_\alpha(\alpha,\beta,\delta)
    \lesssim
    u(\delta)
    \begin{cases}
    \frac{\ent(\alpha)}{\hel(\alpha,\beta)},&
    \beta<\alpha\leq1/2,
    \vspace{0.1cm}\\
    \frac{\alb}{\hel(\alpha,\beta)},&
    \beta<1/2<\alpha,
    \vspace{0.1cm}\\
    \frac{\alb}{\betb}\frac{\ent(\beta)}{\hel(\alpha,\beta)},&
    1/2\leq\beta<\alpha,
    \end{cases}
    \quad
    \icht_\beta(\alpha,\beta,\delta)
    \lesssim
    u(\delta)
    \begin{cases}
    \frac{\beta}{\alpha}\frac{\ent(\alpha)}{\hel(\alpha,\beta)},&
    \beta<\alpha\leq1/2,
    \vspace{0.1cm}\\
    \frac{\beta}{\hel(\alpha,\beta)},&
    \beta<1/2<\alpha,
    \vspace{0.1cm}\\
    \frac{\ent(\beta)}{\hel(\alb,\betb)},&
    1/2\leq\beta<\alpha,
    \end{cases}
\end{align*}
where one may take
\begin{align*}
    u(\delta)
    =
    \begin{cases}
    \delta,&0<\delta<1/3,\\
    \log\frac{1}{1-\delta},&1/3\leq\delta<1.
    \end{cases}
\end{align*}
The implied constants in $\gtrsim$ and $\lesssim$ are universal. For any fixed constant advantage, such as $\delta=1/3$, these lower and upper bounds match up to universal constant factors.
\end{theorem}

\begin{proof}[Proof sketch]The lower bounds follow the framework from \Cref{subsec:coin-lower-bounds}. For a fixed coordinate $i$, the transcript distribution can be viewed as the output of a channel $C\in\cC(\{0,1\},\nats)$ applied to the input bit. Thus the key quantity is
\begin{align*}
    \inf_{C\in\cC(\{0,1\},\nats)}
    \frac{\mi_\alpha(Z;C(Z))}
    {\hel(C(\alpha),C(\beta))}
    =
    \inf_{C\in\cC(\{0,1\},\nats)}
    \frac{\js_\alpha(C(1),C(0))}
    {\hel\lp\alpha C(1)+\alb C(0),\beta C(1)+\betb C(0)\rp}.
\end{align*}
The mixed Hellinger--Jensen--Shannon inequality in \Cref{subsec:coin-mixed-hellinger-js} identifies the optimal ratio in the three regimes.
Combining this with the Hellinger cut-and-paste lemma from \Cref{lem:cut-paste} yields the desired lower bounds.

It is enough to prove the bounds for $\icht_\alpha$ in the three regimes above. The corresponding bounds for $\icht_\beta$ follow by relabelling the two outcomes: after flipping all bits, $\alpha$ and $\beta$ are replaced by $\alb$ and $\betb$, and the information cost under one hypothesis is converted into the information cost under the relabelled hypothesis. The upper bounds are constructive and are proved in \Cref{subsec:coin-upper-bounds}. In the regime $\beta<\alpha\leq1/2$, the optimal protocol is the identity protocol: enough parties simply broadcast their samples. When $\beta<1/2<\alpha$, the identity protocol can reveal too much information, and the matching protocol instead sends very noisy versions of the input bits through a binary symmetric channel. In the final regime $1/2\leq\beta<\alpha$, the matching protocol uses an asymmetric $Z$-channel. These protocols show that the lower bounds are tight and that the optimal channel depends on the location of $\alpha$ and $\beta$ in $[0,1]$.
\end{proof}

\begin{remark}
Note that in regimes 1 and 3, the channels that minimise $\icht_\alpha$ and $\icht_\beta$ are different. Thus, in general, one cannot minimise both information costs simultaneously. This is fine for communication complexity lower bounds, where we can simply use the information cost under whichever hypothesis gives the larger lower bound.
\end{remark}

\subsection{The mixed Hellinger--Jensen--Shannon inequality}\label{subsec:coin-mixed-hellinger-js}

We now state the main $f$-divergence inequality used for the lower bound argument. Fix a binary-input channel $C\in\cC(\{0,1\},\nats)$ and let $Z\sim\ber(\alpha)$. The information revealed by the channel is
\begin{align*}
    \mi_\alpha(Z;C(Z))=\js_\alpha(C(1),C(0)).
\end{align*}
On the other hand, changing the input law from $\ber(\alpha)$ to $\ber(\beta)$ changes the output distribution from $\alpha C(1)+\alb C(0)$ to $\beta C(1)+\betb C(0)$. Thus the lower bound reduces to comparing a skewed Jensen--Shannon divergence with a mixed Hellinger divergence.

\begin{theorem}[Mixed Hellinger--Jensen--Shannon inequality]\label{thm:mixed-hellinger-js}
For any $\alpha\in[0,1]$, any $\beta\in[0,1/2]$, $\beta \neq \alpha$, and any two distributions $\cP$ and $\cQ$ on a discrete space $\cX$,
\begin{align*}
    \js_\alpha(\cP,\cQ)
    \geq
    \frac{\alpha^2\log(1/\alpha)}{(\alpha-\beta)^2}
    \hel(\alpha \cP+\alb \cQ,\beta \cP+\betb \cQ).
\end{align*}
\end{theorem}

The proof has two main steps. First, the mixed Hellinger term on the right-hand side is itself an $f$-divergence between $\cP$ and $\cQ$. Since $\js_\alpha(\cP,\cQ)$ is also an $f$-divergence, the joint-range reduction for pairs of $f$-divergences from \cite{harremoes2011pairs} reduces the theorem to the case $\cP=\ber(p)$ and $\cQ=\ber(q)$. In this binary reduction, it remains to show that
\begin{align*}
    f(\alpha,\beta,p,q)
    \triangleq
    \js_\alpha(p,q)
    -
    \frac{\alpha^2\log(1/\alpha)}{(\alpha-\beta)^2}
    \hel(\alpha p+\alb q,\beta p+\betb q)
    \geq 0.
\end{align*}
For fixed $\alpha,p,q$, the key structural fact is that $f(\alpha,\beta,p,q)$ is concave in $\beta$. It is then enough to verify non-negativity at the relevant boundary points: the limiting point $\beta=\alpha$, and the endpoints $\beta=0$ and $\beta=1/2$. The proof details are deferred to \Cref{app:mixed-hellinger-js}.

\subsection{Proofs of \Cref{thm:coin-info-section3} lower bounds}\label{subsec:coin-lower-bounds}

We now prove the lower-bound half of \Cref{thm:coin-info-section3}. Let $\Pi$ be an arbitrary public-coin blackboard protocol using $n$ samples and distinguishing $\ber(\alpha)$ from $\ber(\beta)$ with advantage at least $\delta$.

We first prove the two regimes in which $\beta<1/2$. The independence of the samples and \Cref{ineq:mi-sup-add} give
\begin{align*}
    \mi_\alpha(X;\Pi)
    \geq
    \sum_{i\in[n]}\mi_\alpha(X_i;\Pi).
\end{align*}
For each coordinate $i$, the transcript is the output of a binary-input channel whose two input-conditioned laws are $\Pi_\alpha(X\subt{i}1)$ and $\Pi_\alpha(X\subt{i}0)$. Hence
\begin{align*}
    \mi_\alpha(X_i;\Pi)
    =
    \js_\alpha\lp\Pi_\alpha(X\subt{i}1),\Pi_\alpha(X\subt{i}0)\rp.
\end{align*}
Since $\beta<1/2$, \Cref{thm:mixed-hellinger-js} applies and yields
\begin{align*}
    \mi_\alpha(X_i;\Pi)
    \geq
    \frac{\alpha^2\log(1/\alpha)}{(\alpha-\beta)^2}
    \hel\lp\Pi_\alpha(X),\Pi_\alpha(X\subt{i}\beta)\rp.
\end{align*}
Summing over $i$, using the cut-and-paste inequality \Cref{lem:cut-paste}, and then using \Cref{fact:tv-hel-ineq} together with the advantage assumption gives
\begin{align}\label{eq:alpha-lower-bound-beta-less-half}
    \mi_\alpha(X;\Pi)
    \geq
    \frac{\alpha^2\log(1/\alpha)}{(\alpha-\beta)^2}
    \sum_{i\in[n]}\hel\lp\Pi_\alpha(X),\Pi_\alpha(X\subt{i}\beta)\rp
    \gtrsim
    \frac{\alpha^2\log(1/\alpha)}{(\alpha-\beta)^2}\delta^2.
\end{align}

When $\beta<\alpha\leq1/2$, \Cref{fact:hel_ber} gives $\hel(\alpha,\beta)\asymp(\alpha-\beta)^2/\alpha$, and $\ent(\alpha)\asymp\alpha\log(1/\alpha)$. Therefore \Cref{eq:alpha-lower-bound-beta-less-half} implies
\begin{align*}
    \mi_\alpha(X;\Pi)
    \gtrsim
    \delta^2\frac{\ent(\alpha)}{\hel(\alpha,\beta)}.
\end{align*}
When $\beta<1/2<\alpha$, the same bound \Cref{eq:alpha-lower-bound-beta-less-half} applies. In this regime $\hel(\alpha,\beta)\asymp(\alpha-\beta)^2$, while $\alpha^2\log(1/\alpha)\asymp \log(1/\alpha) \asymp \alb$ for $\alpha>1/2$. Hence
\begin{align*}
    \mi_\alpha(X;\Pi)
    \gtrsim
    \delta^2\frac{\alb}{\hel(\alpha,\beta)}.
\end{align*}

It remains to prove the lower bound for $\mi_\alpha(X;\Pi)$ in the third case, where $1/2\leq\beta<\alpha$. Let
\begin{align*}
    \alpha'=\alb,
    \qquad
    \beta'=\betb,
\end{align*}
so that $0\leq\alpha'<\beta'\leq1/2$. For each coordinate $i$, write
\begin{align*}
    \cP_i=\Pi_\alpha(X\subt{i}1),
    \qquad
    \cQ_i=\Pi_\alpha(X\subt{i}0).
\end{align*}
Using the symmetry of the skewed Jensen--Shannon divergence,
\begin{align*}
    \mi_\alpha(X_i;\Pi)
    =
    \js_\alpha(\cP_i,\cQ_i)
    =
    \js_{\alpha'}(\cQ_i,\cP_i).
\end{align*}
Since $\js_t(\cQ_i,\cP_i)$ is the mutual information of a fixed binary-input channel with input law $\ber(t)$, it is concave in $t$. Therefore,
\begin{align*}
    \js_{\alpha'}(\cQ_i,\cP_i)
    \geq
    \frac{\alpha'}{\beta'}\js_{\beta'}(\cQ_i,\cP_i).
\end{align*}
We now apply \Cref{thm:mixed-hellinger-js} with theorem parameters $\beta'$ and $\alpha'$, and with the two distributions $\cQ_i,\cP_i$. This gives
\begin{align*}
    \mi_\alpha(X_i;\Pi)
    &\geq
    \frac{\alpha'}{\beta'}
    \cdot
    \frac{(\beta')^2\log(1/\beta')}{(\beta'-\alpha')^2}
    \hel\lp
        \beta'\cQ_i+\bar{\beta'}\cP_i,
        \alpha'\cQ_i+\bar{\alpha'}\cP_i
    \rp
    \\
    &=
    \frac{\alb\betb\log(1/\betb)}{(\alpha-\beta)^2}
    \hel\lp
        \Pi_\alpha(X\subt{i}\beta),
        \Pi_\alpha(X)
    \rp .
\end{align*}
Summing over $i$ and using \Cref{lem:cut-paste,fact:tv-hel-ineq} yields
\begin{align*}
    \mi_\alpha(X;\Pi)
    \gtrsim
    \delta^2
    \frac{\alb\betb\log(1/\betb)}{(\alpha-\beta)^2}.
\end{align*}
Finally, $\ent(\beta)=\ent(\betb)\asymp\betb\log(1/\betb)$ and, by \Cref{fact:hel_ber},
\begin{align*}
    \hel(\alpha,\beta)
    =
    \hel(\alb,\betb)
    \asymp
    \frac{(\alpha-\beta)^2}{\betb}.
\end{align*}
Thus
\begin{align*}
    \mi_\alpha(X;\Pi)
    \gtrsim
    \delta^2
    \frac{\alb}{\betb}
    \frac{\ent(\beta)}{\hel(\alpha,\beta)},
\end{align*}
as required.

It remains only to explain the lower bounds for $\mi_\beta(X;\Pi)$. Relabel the outcomes by replacing each bit with its complement, and again write
\begin{align*}
    \alpha'=\alb,
    \qquad
    \beta'=\betb.
\end{align*}
The original information cost under $\ber(\beta)$ becomes the information cost under $\ber(\beta')$ for the flipped problem of testing $\ber(\beta')$ against $\ber(\alpha')$, where $\alpha'<\beta'$. Thus we apply the preceding $\mi_\alpha$ lower bounds to the ordered pair $(\beta',\alpha')$. Under this relabelling, the regimes transform as
\begin{align*}
    \beta<\alpha\leq1/2
    &\longmapsto
    1/2\leq\alpha'<\beta',
    \\
    \beta<1/2<\alpha
    &\longmapsto
    \alpha'<1/2<\beta',
    \\
    1/2\leq\beta<\alpha
    &\longmapsto
    \alpha'<\beta'\leq1/2.
\end{align*}
Therefore the first, second, and third $\mi_\beta$ bounds follow respectively from the third, second, and first $\mi_\alpha$ bounds applied to the flipped problem.

\subsection{Proofs of \Cref{thm:coin-info-section3} upper bounds}\label{subsec:coin-upper-bounds}

We now prove the upper bounds for $\mi_\alpha(X;\Pi)$ in the three regimes of \Cref{thm:coin-info-section3}. All protocols used here are non-interactive: each selected party applies the same binary channel $C$ to its sample and writes the resulting bit on the blackboard. If $C(\alpha)$ and $C(\beta)$ denote the output laws of this channel under $\ber(\alpha)$ and $\ber(\beta)$, we use
\begin{align*}
    n
    =
    n^*(C(\alpha),C(\beta),\delta)
\end{align*}
agents and then run the optimal centralised test on their messages. By \Cref{fact:hyp_sample_comp_fact}, this contributes the factor $u(\delta)/\hel(C(\alpha),C(\beta))$ appearing in the upper bounds.

First consider the regime $\beta<\alpha\leq1/2$. Use the identity channel. Then $C(\alpha)=\ber(\alpha)$, $C(\beta)=\ber(\beta)$, and each message reveals $\ent(\alpha)$ information under $\ber(\alpha)$. Hence
\begin{align*}
    \mi_\alpha(X;\Pi)
    \lesssim
    u(\delta)\frac{\ent(\alpha)}{\hel(\alpha,\beta)}.
\end{align*}

Next consider the regime $\beta<1/2<\alpha$. Let $C_\ve$ be the binary symmetric channel with crossover probability $1/2-\ve$, where $\ve>0$ is fixed later. Then
\begin{align*}
    C_\ve(\alpha)
    =
    \ber\lp1/2+(2\alpha-1)\ve\rp,
    \qquad
    C_\ve(\beta)
    =
    \ber\lp1/2+(2\beta-1)\ve\rp.
\end{align*}
The estimates in \Cref{app:coin-upper-bound-channel-calculations} give
\begin{align*}
    \hel(C_\ve(\alpha),C_\ve(\beta))
    \asymp
    \ve^2(\alpha-\beta)^2,
    \qquad
    \mi_\alpha(X_i;C_\ve(X_i))
    \asymp
    \ve^2\alb.
\end{align*}
Taking $\ve$ to be a sufficiently small absolute constant, and using $\hel(\alpha,\beta)\asymp(\alpha-\beta)^2$ in this regime, yields
\begin{align*}
    \mi_\alpha(X;\Pi)
    \lesssim
    u(\delta)\frac{\alb}{\hel(\alpha,\beta)}.
\end{align*}

Finally suppose $1/2\leq\beta<\alpha$. Here we use the asymmetric Z-channel $\cZ_\beta$ defined by
\begin{align*}
    \cZ_\beta(0)
    &\sim
    \ber(0),
    &
    \cZ_\beta(1)
    &\sim
    \ber(\beta).
\end{align*}
Thus, the channel outputs $1$ only when its input is $1$ and an independent $\ber(\beta)$ coin is also $1$. In particular, under an input distributed as $\ber(\gamma)$, its output law is
\begin{align*}
    \cZ_\beta(\gamma)
    =
    \ber(\beta\gamma).
\end{align*}
Consequently, the two hypotheses induce the output laws $\ber(\alpha\beta)$ and $\ber(\beta^2)$. Since $\alpha-\beta<\betb$, \Cref{fact:hel_ber} gives
\begin{align*}
    \hel(\cZ_\beta(\alpha),\cZ_\beta(\beta))
    \asymp
    \frac{(\alpha-\beta)^2}{\betb}.
\end{align*}
Moreover, one use of the channel reveals
\begin{align*}
    \mi_\alpha(X_i;\cZ_\beta(X_i))
    =
    \ent(\alpha\beta)-\alpha\ent(\beta)
    \lesssim
    \alb\log(1/\betb)
\end{align*}
information under $\ber(\alpha)$. This is calculated in \Cref{app:coin-upper-bound-channel-calculations}. We therefore use
\begin{align*}
    n
    \asymp
    u(\delta)\frac{\betb}{(\alpha-\beta)^2}
\end{align*}
agents, each applying $\cZ_\beta$ to its sample. The resulting protocol satisfies
\begin{align*}
    \mi_\alpha(X;\Pi)
    \lesssim
    u(\delta)\frac{\alb\betb\log(1/\betb)}{(\alpha-\beta)^2}
    \asymp
    u(\delta)
    \frac{\alb}{\betb}
    \frac{\ent(\beta)}{\hel(\alpha,\beta)}.
\end{align*}

The upper bounds for $\mi_\beta(X;\Pi)$ follow by the same relabelling symmetry used for the lower bounds. 

\subsection{Comparison with prior lower-bound methods}\label{subsec:coin-comparison}

We start by comparing our bounds to the coin problem lower bound presented in \cite{streaming2025} for the special case of $\beta = 0$. While the paper primarily studies information complexity lower bounds in the streaming setting, their analysis depends crucially on a lower bound for the broadcast model. We restate this bound using our notation below.

\begin{theorem}[Strengthening of Theorem~15 in \cite{streaming2025}]\label{thm:colt25}
    Let $0 = \beta < \alpha \le 1$. The information complexity of the coin problem of testing $\ber(\alpha)$ versus $\ber(0)$ under $\ber(\alpha)$ satisfies 
        \begin{align*}
        \icht_{\alpha}(\delta)
        \gtrsim \min\{1/2, \bar \alpha\}\delta^2.
    \end{align*}
\end{theorem}

\begin{remark}
The bound in \cite{streaming2025} is stated as $\min\{\alpha,\alb\}\delta^2$. A simple modification of their proof, or an application of the data-processing inequality, gives the stronger bound in \Cref{thm:colt25}.
\end{remark}

For comparison, fix $\delta=1/3$. When $\alpha\leq1/2$, \Cref{thm:colt25} gives only the constant lower bound $\icht_\alpha\gtrsim1$, whereas \Cref{thm:coin-info-section3} gives the sharp characterisation
\begin{align*}
    \icht_\alpha
    \asymp
    \log(1/\alpha).
\end{align*}
The additional $\log(1/\alpha)$ factor captures the increase in information complexity as the testing problem becomes harder when $\alpha\to0$. This dependence is important in the applications in \Cref{sec:applications}. When $1/2<\alpha$, the lower bound in \Cref{thm:colt25} has the correct order $\alb$, matching the sharp characterisation in \Cref{thm:coin-info-section3}.

The loss for $\alpha\leq1/2$ comes from the two general-purpose inequalities used in the earlier argument. First, concavity of mutual information in the input distribution gives
\begin{align*}
    \mi_\alpha(X_i;\Pi)
    \gtrsim
    \alpha\mi_{1/2}(X_i;\Pi)
    \gtrsim
    \alpha\hel\lp
        \Pi_\alpha(X\subt{i}0),
        \Pi_\alpha(X\subt{i}1)
    \rp.
\end{align*}
Second, joint convexity of squared Hellinger divergence gives
\begin{align*}
    \hel\lp
        \Pi_\alpha(X\subt{i}0),
        \Pi_\alpha(X\subt{i}1)
    \rp
    \gtrsim
    \frac{1}{\alpha}
    \hel\lp
        \Pi_\alpha(X),
        \Pi_\alpha(X\subt{i}0)
    \rp.
\end{align*}
Together these yield only
\begin{align*}
    \mi_\alpha(X_i;\Pi)
    \gtrsim
    \hel\lp
        \Pi_\alpha(X),
        \Pi_\alpha(X\subt{i}0)
    \rp.
\end{align*}
By contrast, applying \Cref{thm:mixed-hellinger-js} directly gives
\begin{align*}
    \mi_\alpha(X_i;\Pi)
    \gtrsim
    \log(1/\alpha)
    \hel\lp
        \Pi_\alpha(X),
        \Pi_\alpha(X\subt{i}0)
    \rp,
\end{align*}
which recovers the missing logarithmic factor.

Next, we compare our results to those obtained in \cite{braverman2016communication}. To state the lower bound from that paper, we need the following definition.
\begin{definition}[Post-SDPI]\label{def:post_sdpi}
    Let $P_{Y|X}$ be a channel, let $X \sim P_X$ be the channel input, and $Y$ be the channel output. The input-dependent post-data processing constant (post-SDPI constant) is defined as
    \begin{align*}
        \eta(P_{Y|X},P_X)=\sup_{Z: X \to Y \to Z}\frac{\mi(X;Z)}{\mi(Y;Z)},
    \end{align*}
    where the supremum is over all Markov chains $X \to Y \to Z$.
\end{definition}
Note that $\mi(X; Z) \leq \mi(Y;Z)$ by the usual data processing inequality, so the post-SDPI constant is at most 1. When it is strictly smaller, we have the strong data processing inequality $\mi(X;Z) \leq \eta \mi(Y;Z)$ where $\eta = \eta(P_{Y|X}, P_X)$. The information complexity lower bound from \cite{braverman2016communication} stated in our notation is as follows:

\begin{theorem}[Theorem 1.1, \cite{braverman2016communication}]\label{thm:braverman_sdpi}
Let $V \sim \ber(1/2)$ and let $X \sim \cP$ when $V=1$ and $X \sim \cQ$ when $V = 0$. Let $\eta = \eta(P_{X|V}, P_V)$ be the post-SDPI constant and let $L = \max_x \left\{\frac{\cP(x)}{\cQ(x)}, \frac{\cQ(x)}{\cP(x)} \right\}$, which is assumed to be finite. Then
\begin{align*}
\min\{\icht_\cP(\delta), \icht_\cQ(\delta)\} \gtrsim \frac{\delta^2}{L\eta}.
\end{align*}
\end{theorem}
The proof of this result uses the same framework as in \Cref{subsection:lower-bound-framework}, but with a slight modification. Let $V \sim \ber(1/2)$, and generate $Z_j \sim \cP$ if $V=1$ and $Z_j \sim \cQ$ if $V=0$. All other $Z_i$'s for $i \in [n]$ and $i \neq j$ are i.i.d.\ according to $\cP$. We have the Markov chain $V \to Z_j \to \Pi$, giving the inequality
\begin{align*}
\mi(V; \Pi) \leq \eta \mi(Z_j; \Pi).
\end{align*}
Since $V\sim\ber(1/2)$, mutual information is the ordinary Jensen--Shannon divergence between the two conditional transcript laws and is therefore comparable to their Hellinger divergence. Hence
\begin{align*}
    \mi(V;\Pi)
    \asymp
    \hel\lp\Pi_\cP(X),\Pi_\cP(X\subt{j}\cQ)\rp.
\end{align*}
The likelihood-ratio bound yields
\begin{align*}
    \mi(Z_j;\Pi)
    \lesssim
    L \, \mi_\cP(X_j;\Pi).
\end{align*}
Combining the two estimates gives
\begin{align*}
    \hel\lp\Pi_\cP(X),\Pi_\cP(X\subt{j}\cQ)\rp
    \lesssim
    \eta L  \,\mi_\cP(X_j;\Pi).
\end{align*}
As this is true for all $j \in [n]$, we can use \Cref{lem:cut-paste} and conclude the bound in \Cref{thm:braverman_sdpi} for $\icht_\cP(\delta)$. Interchanging the roles of $\cP$ and $\cQ$ in the above proof yields the same bound for $\icht_\cQ(\delta)$.

We now specialise the post-SDPI bound to Bernoulli testing in the following lemma, whose proof is deferred to \Cref{app:bernoulli-post-sdpi}.

\begin{lemma}[Post-SDPI constant for Bernoulli testing]\label{lem:psdpi_ber}
Let $0<\beta<\alpha<1$. Let $X\sim\ber(1/2)$, and conditionally let $Y\sim\ber(\alpha)$ when $X=1$ and $Y\sim\ber(\beta)$ when $X=0$. Then
\begin{align*}
    \eta(P_{Y|X},\ber(1/2))
    \asymp
    \frac{(\alpha-\beta)^2}
    {(\alpha+\beta)^2(\alb+\betb)^2
    \log\lp\dfrac{2}{\min\{\alpha+\beta,\alb+\betb\}}\rp}.
\end{align*}
\end{lemma}

Assume first that $0<\beta<\alpha<1$, so that the likelihood-ratio constant is finite. For $\cP=\ber(\alpha)$ and $\cQ=\ber(\beta)$, it is
\begin{align*}
    L
    &=
    \max\left\{
        \frac{\alpha}{\beta},
        \frac{\betb}{\alb}
    \right\}
    =
    \begin{cases}
        \alpha/\beta,&\alpha+\beta\leq1,\\
        \betb/\alb,&\alpha+\beta\geq1.
    \end{cases}
\end{align*}
Indeed, comparing the two candidate ratios reduces to the sign of $(\alpha-\beta)(1-\alpha-\beta)$. Substituting this value of $L$ and the estimate from \Cref{lem:psdpi_ber} into \Cref{thm:braverman_sdpi} gives the following comparison in the three regimes of \Cref{thm:coin-info-section3}.

\paragraph{Regime 1 ($0 < \beta < \alpha \leq 1/2$).}
Here, $\alpha+\beta\asymp\alpha$, $\alb+\betb\asymp1$, and $L=\alpha/\beta$. Therefore,
\begin{align*}
    \frac{1}{L\eta}
    \asymp
    \frac{\alpha\beta\log(1/\alpha)}{(\alpha-\beta)^2}
    \asymp
    \frac{\beta}{\alpha}
    \frac{\ent(\alpha)}{\hel(\alpha,\beta)}.
\end{align*}
Thus the post-SDPI method recovers the sharp lower bound for $\icht_\beta$, which is the smaller of the two information complexities in this regime. It can miss the sharp lower bound for $\icht_\alpha$ by a factor of order $\alpha/\beta$.

\paragraph{Regime 2 ($\beta < 1/2 < \alpha$).}
In this case $\alpha+\beta$, $\alb+\betb$, and the logarithmic term in \Cref{lem:psdpi_ber} are all bounded above and below by universal positive constants. Hence $\eta\asymp(\alpha-\beta)^2$. If $\alpha+\beta\leq1$, equivalently $\beta\leq\alb$, then $L=\alpha/\beta\asymp1/\beta$ and
\begin{align*}
    \frac{1}{L\eta}
    \asymp
    \frac{\beta}{(\alpha-\beta)^2}.
\end{align*}
This matches the sharp bound for $\icht_\beta$. If $\alpha+\beta\geq1$, equivalently $\alb\leq\beta$, then $L=\betb/\alb\asymp1/\alb$ and
\begin{align*}
    \frac{1}{L\eta}
    \asymp
    \frac{\alb}{(\alpha-\beta)^2},
\end{align*}
which matches the sharp bound for $\icht_\alpha$. Thus the post-SDPI method again recovers precisely the smaller of the two information complexities.

\paragraph{Regime 3 ($1/2\leq\beta<\alpha$).}
Here $\alb+\betb\asymp\betb$, $\alpha+\beta\asymp1$, and $L=\betb/\alb$. It follows that
\begin{align*}
    \frac{1}{L\eta}
    \asymp
    \frac{\alb\betb\log(1/\betb)}{(\alpha-\beta)^2}
    \asymp
    \frac{\alb}{\betb}
    \frac{\ent(\beta)}{\hel(\alpha,\beta)}.
\end{align*}
This is the sharp lower bound for $\icht_\alpha$, the smaller information complexity in this regime, but it can miss the sharp $\icht_\beta$ bound by a factor of order $\betb/\alb$.

In summary, whenever $L$ is finite, the post-SDPI method recovers the smaller of $\icht_\alpha$ and $\icht_\beta$ up to universal constant factors in the constant-advantage regime, but it does not generally recover the larger information cost. At the boundary cases $\beta=0$ or $\alpha=1$, the likelihood-ratio constant is infinite and the method gives no non-trivial bound. By contrast, the mixed Hellinger--Jensen--Shannon inequality applies directly under either input distribution, yields the sharp information cost under both hypotheses, and remains meaningful at these boundary cases.
\section{Information complexity of $\cP$ versus $\cQ$}\label{sec:general-information}

In this section, we study the information complexity of testing two arbitrary distinct distributions $\cP$ and $\cQ$ on $[k]$. We first give a constant-factor characterisation in terms of an optimisation over channels and show that binary output channels suffice for this purpose. We then use the coin-problem characterisation to prove lower bounds for distributions with bounded likelihood ratio and identify a regime in which these bounds are sharp. Finally, we construct a noisy binary likelihood-ratio channel that gives general upper bounds in terms of the $\chi^2$ divergences.

\subsection{Information complexity via optimisation over channels}\label{subsec:general-channel-characterisation}

Let $\cP$ and $\cQ$ be distinct distributions on $[k]$. We begin by identifying the channel optimisation that characterises their information complexity. For $Z\sim\cP$, define
\begin{align}\label{eq:def-phi-pq}
    \phi_{\cP}(\cP,\cQ)
    \coloneqq
    \inf_{\substack{C\in\cC([k],\nats)\\
    \hel(C\cP,C\cQ)>0}}
    \frac{\mi_{\cP}(Z;C(Z))}{\hel(C\cP,C\cQ)}.
\end{align}
Note that $\phi_{\cP}(\cP, \cQ)$ need not equal $\phi_\cQ(\cP, \cQ)$ in general.

\begin{theorem}[Information complexity characterisation]\label{thm:general-channel-characterisation}
    Let $\cP$ and $\cQ$ be distinct distributions on $[k]$. For advantage $1/3$,
    \begin{align*}
        \icht_{\cP}(\cP,\cQ,1/3)
        \asymp
        \phi_{\cP}(\cP,\cQ),
    \end{align*}
    where the implicit constants are universal.
\end{theorem}

\begin{proof}
    We first prove the lower bound by applying the framework from \Cref{subsection:lower-bound-framework}. Let $\Pi$ be an $n$-party public-coin blackboard protocol that distinguishes $\cP^{\otimes n}$ from $\cQ^{\otimes n}$ with advantage at least $1/3$. For each $i\in[n]$, let $C_i\in\cC([k],\nats)$ be the channel that maps a fixed value of $X_i$ to the distribution of the transcript, where the remaining coordinates are sampled independently from $\cP$. Thus, $C_i\cP=\Pi_{\cP}(X)$ and $C_i\cQ=\Pi_{\cP}(X\subt{i}\cQ)$. Additionally, $\mi_{\cP}(X_i;\Pi) = \mi_{\cP}(Z;C_i(Z))$. By the definition of $\phi_{\cP}(\cP,\cQ)$,
    \begin{align*}
        \mi_{\cP}(X_i;\Pi)
        \geq
        \phi_{\cP}(\cP,\cQ)
        \hel\lp
            \Pi_{\cP}(X),
            \Pi_{\cP}(X\subt{i}\cQ)
        \rp.
    \end{align*}
Summing over $i$, and then applying \Cref{ineq:mi-sup-add,ineq:sum-hel-lb}, gives
    \begin{align*}
        \mi_{\cP}(X;\Pi)
        &\geq
        \sum_{i\in[n]}\mi_{\cP}(X_i;\Pi)
        \\
        &\geq
        \phi_{\cP}(\cP,\cQ)
        \sum_{i\in[n]}
        \hel\lp
            \Pi_{\cP}(X),
            \Pi_{\cP}(X\subt{i}\cQ)
        \rp
        \\
        &\gtrsim
        \phi_{\cP}(\cP,\cQ).
    \end{align*}
    Taking the infimum over $\Pi$ proves the lower bound.

    For the upper bound, first suppose that $\phi_{\cP}(\cP,\cQ)>0$, and choose a channel $C$ such that

  \begin{align*}
        \frac{\mi_{\cP}(Z;C(Z))}{\hel(C\cP,C\cQ)}
        \leq 2
        \phi_{\cP}(\cP,\cQ).
    \end{align*}
Let $n=n^*(C\cP,C\cQ,1/3)$. Each of the $n$ agents independently applies $C$ to their sample and writes the output on the shared blackboard. By \Cref{fact:hyp_sample_comp_fact}, $n\lesssim1/\hel(C\cP,C\cQ)$ when $\hel(C\cP,C\cQ)\leq1/2$; when $\hel(C\cP,C\cQ)>1/2$, one sample suffices. Thus, in either case, the protocol has advantage at least $1/3$ and its information cost under $\cP$ is
    \begin{align*}
        \mi_{\cP}(X;C(X_1),\dots,C(X_n))
        &=
        n\mi_{\cP}(Z;C(Z))
        \lesssim
        \frac{\mi_{\cP}(Z;C(Z))}{\hel(C\cP,C\cQ)}
        \lesssim
        2\phi_{\cP}(\cP,\cQ).
    \end{align*}
    If $\phi_{\cP}(\cP,\cQ)=0$, apply the same construction to a sequence of channels whose ratios tend to zero. The resulting information costs tend to zero, so $\icht_{\cP}(\cP,\cQ,1/3)=0$.
\end{proof}

\begin{remark}\label{rem:interactivity}
    The upper-bound protocol in \Cref{thm:general-channel-characterisation} is non-interactive as every agent independently applies the same channel and broadcasts the result. Consequently, interaction cannot improve the information complexity by more than a universal constant. The corresponding conclusion for sample complexity or communication complexity of blackboard protocols does not follow from this result.
\end{remark}

\subsection{Reduction to binary output channels}\label{subsec:binary-output-suffices}

We next show that the optimisation in \Cref{eq:def-phi-pq} may be restricted to binary output channels. Define
\begin{align*}
    \phi^{(2)}_{\cP}(\cP,\cQ)
    \coloneqq
    \inf_{\substack{B\in\cC([k],[2])\\
    \hel(B\cP,B\cQ)>0}}
    \frac{\mi_{\cP}(Z;B(Z))}{\hel(B\cP,B\cQ)}.
\end{align*}

\begin{theorem}[Binary output channels are sufficient]\label{thm:binary-output-suffices}
    Let $\cP$ and $\cQ$ be distinct distributions on $[k]$. Then
    \begin{align*}
        \phi_{\cP}(\cP,\cQ)
        =
        \phi^{(2)}_{\cP}(\cP,\cQ).
    \end{align*}
    Consequently,
    \begin{align*}
        \icht_{\cP}(\cP,\cQ,1/3)
        \asymp
        \inf_{\substack{B\in\cC([k],[2])\\
        \hel(B\cP,B\cQ)>0}}
        \frac{\mi_{\cP}(Z;B(Z))}{\hel(B\cP,B\cQ)}.
    \end{align*}
\end{theorem}

\begin{proof}
From $\cC([k],[2])\subseteq\cC([k],\nats)$, we have $\phi_{\cP}(\cP,\cQ) \leq \phi^{(2)}_{\cP}(\cP,\cQ)$. We now prove the reverse inequality. First suppose that $\phi_{\cP}(\cP,\cQ)>0$, and choose any channel $C\in\cC([k],[m])$. Consider the ratio in \Cref{eq:def-phi-pq} for $C$. We show that there exists a sequence of binary output channels such that, in the limit, the ratio in \Cref{eq:def-phi-pq} corresponding to these channels approaches the value of the ratio for $C$.  

    For each $y\in[m]$, write
    \begin{align*}
        a_y&=(C\cP)(y),
        &b_y&=(C\cQ)(y),
        \\
        \iota_y
        &\coloneqq
        \sum_{x\in[k]}\cP(x)C(y\mid x)
        \log\frac{C(y\mid x)}{a_y},
        &
        h_y
        &\coloneqq
        \frac{1}{2}\lp\sqrt{a_y}-\sqrt{b_y}\rp^2,
    \end{align*}
    with the usual convention for zero terms. These quantities decompose the numerator and denominator as
    \begin{align*}
        \mi_{\cP}(Z;C(Z))
        =
        \sum_{y\in[m]}\iota_y,
        \qquad
        \hel(C\cP,C\cQ)
        =
        \sum_{y\in[m]}h_y.
    \end{align*}
    Therefore, there exists $y\in[m]$ with $h_y>0$ such that
    \begin{align}\label{eq:good-output-symbol}
        \frac{\iota_y}{h_y}
        \leq
        \frac{\mi_{\cP}(Z;C(Z))}{\hel(C\cP,C\cQ)}.
    \end{align}

    Fix such a $y$. For $t\in(0,1]$, define a binary channel $B_t:[k]\to\{0,1\}$ by
    \begin{align*}
        B_t(1\mid x)=tC(y\mid x),
        \qquad
        B_t(0\mid x)=1-tC(y\mid x).
    \end{align*}
   As $t\to0$,
    \begin{align*}
        \mi_{\cP}(Z;B_t(Z))
        &=t\iota_y+O(t^2),
        \\
        \hel(B_t\cP,B_t\cQ)
        &=t h_y+O(t^2).
    \end{align*}
    Indeed, the contribution of the output $1$ to the mutual information is exactly $t\iota_y$, while the contribution of the output $0$ is $O(t^2)$. Similarly, the contribution of the output $1$ to the Hellinger divergence is exactly $t h_y$, while that of the output $0$ is $O(t^2)$. If $\iota_y>0$, the ratio is $\iota_y/h_y + O(t)$, and so when $t \to 0$, this approaches $\iota_y/h_y$ which satisfies the claimed inequality via \Cref{eq:good-output-symbol}. If $\iota_y=0$, then $C(y\mid x)$ is constant for $x$ in the support of $\cP$, and hence $\mi_{\cP}(Z;B_t(Z))=0$ for every $t$. Finally, if $\phi_{\cP}(\cP,\cQ)=0$, apply the same argument to channels whose ratios approach zero. This shows that $\phi^{(2)}_{\cP}(\cP,\cQ)=0$ and completes the proof.
\end{proof}


\begin{remark}
The binary reduction in \Cref{thm:binary-output-suffices} is reminiscent of the reduction used to prove the mixed Hellinger--Jensen--Shannon inequality in \Cref{sec:coin-information}, but the reason it holds is different. There, both the Jensen--Shannon divergence and the mixed Hellinger divergence are $f$-divergences between two distributions, so the joint-range theorem of Harremoës and Vajda~\cite{harremoes2011pairs} reduces the problem to distributions on a binary alphabet. Here the input alphabet has size $k$, and $\mi_{\cP}(Z;C(Z))$ depends jointly on the $k$ conditional output distributions of the channel; it is not an $f$-divergence between a single pair of distributions. The Harremoës--Vajda reduction therefore does not apply directly. 

This viewpoint also describes the \emph{a posteriori} distributions induced by the matching channels for the three regimes of the coin problem. Let $X\sim\ber(\alpha)$ and let $Y$ be the channel output. When $\beta<\alpha\leq1/2$, the identity channel is optimal up to constants, and observing the informative output $Y=1$ gives the point-mass posterior $\Pr(X=1\mid Y=1)=1$. When $\beta<1/2<\alpha$, the matching binary symmetric channels have crossover probability $1/2-\ve$; as $\ve\to0$, their posteriors converge to the prior $\ber(\alpha)$, so the useful posterior is an infinitesimal perturbation of the prior. Finally, when $1/2\leq\beta<\alpha$, the matching channel is the asymmetric channel $\cZ_\beta$ from \Cref{subsec:coin-upper-bounds}. For its output $Y=0$, the posterior is
\begin{align*}
    \Pr(X=1\mid Y=0)
    =
    \frac{\alpha(1-\beta)}{1-\alpha\beta}.
\end{align*}
Thus the three regimes correspond respectively to a boundary posterior, a small perturbation of the prior, and a non-trivial asymmetric posterior.
\end{remark}

\subsection{Lower bounds for bounded likelihood ratios}\label{subsec:bounded-likelihood-ratios}

We next use the coin-problem characterisation to derive lower bounds for testing general distributions. The following lemma gives the required reduction to the coin problem.

\begin{lemma}[Bernoulli embedding]\label{lem:general-bernoulli-embedding}
    Let $\cP$ and $\cQ$ be distributions on $[k]$. Suppose that there is a channel $K\in\cC(\{0,1\},[k])$ and parameters $t_{\cP},t_{\cQ}\in[0,1]$ such that
    \begin{align*}
        K\ber(t_{\cP})=\cP,
        \qquad
        K\ber(t_{\cQ})=\cQ.
    \end{align*}
    Then, for each $\cR\in\{\cP,\cQ\}$,
    \begin{align*}
        \phi_{\cR}(\cP,\cQ)
        \geq
        \phi_{\ber(t_{\cR})}
        \lp\ber(t_{\cP}),\ber(t_{\cQ})\rp.
    \end{align*}
    Consequently,
    \begin{align*}
        \icht_{\cR}(\cP,\cQ,1/3)
        \gtrsim
        \icht_{t_{\cR}}(t_{\cP},t_{\cQ},1/3).
    \end{align*}
\end{lemma}

\begin{proof}
    Fix a channel $C\in\cC([k],\nats)$ with $\hel(C\cP,C\cQ)>0$, and let $U\sim\ber(t_{\cR})$, $Z\sim K(U)$, and $Y\sim C(Z)$. Then $Z\sim\cR$ and $U-Z-Y$ is a Markov chain. Hence, by data processing,
    \begin{align*}
        \mi_{\cR}(Z;C(Z))
        \geq
        \mi_{t_{\cR}}(U;(C\circ K)(U)).
    \end{align*}
    Moreover,
    \begin{align*}
        (C\circ K)\ber(t_{\cP})=C\cP,
        \qquad
        (C\circ K)\ber(t_{\cQ})=C\cQ.
    \end{align*}
    Dividing the data-processing inequality by $\hel(C\cP,C\cQ)$ and then taking the infimum over $C$ proves the first claim. The second follows from \Cref{thm:general-channel-characterisation}, applied once to $\cP,\cQ$ and once to the two Bernoulli distributions.
\end{proof}

We now specialise the embedding to distributions with bounded likelihood ratio.

\begin{theorem}[Bounded likelihood ratios]\label{thm:bounded-likelihood-ratio-lower-bound}
    Let $\cP$ and $\cQ$ be distinct distributions on $[k]$, and suppose that, for some $\varepsilon\in(0,1/2]$,
    \begin{align*}
        (1-\varepsilon)\cQ(x)
        \leq
        \cP(x)
        \leq
        (1+\varepsilon)\cQ(x)
        \qquad
        \text{for every $x\in[k]$}.
    \end{align*}
    Then, for each $\cR\in\{\cP,\cQ\}$,
    \begin{align*}
        \icht_{\cR}(\cP,\cQ,1/3)
        \gtrsim
        \frac{1}{\varepsilon^2}.
    \end{align*}
    If, in addition, $\dtv(\cP,\cQ)\asymp\varepsilon$, then this lower bound is tight up to universal constant factors under both hypotheses.
\end{theorem}

\begin{proof}
    Define two distributions on $[k]$ by
    \begin{align*}
        K(1)(x)
        &\coloneqq
        \cQ(x)+\frac{\cP(x)-\cQ(x)}{\varepsilon},
        \\
        K(0)(x)
        &\coloneqq
        \cQ(x)-\frac{\cP(x)-\cQ(x)}{\varepsilon}.
    \end{align*}
    The likelihood-ratio assumption implies that both expressions are non-negative, and each sums to one. They therefore define a channel $K\in\cC(\{0,1\},[k])$. Setting
    \begin{align*}
        t_{\cP}=\frac{1+\varepsilon}{2},
        \qquad
        t_{\cQ}=\frac{1}{2},
    \end{align*}
    a direct calculation gives
    \begin{align*}
        K\ber(t_{\cP})=\cP,
        \qquad
        K\ber(t_{\cQ})=\cQ.
    \end{align*}
    Since both Bernoulli parameters remain bounded away from $0$ and $1$, \Cref{thm:coin-info-section3} gives
    \begin{align*}
        \icht_{t_{\cP}}(t_{\cP},t_{\cQ},1/3)
        \asymp
        \icht_{t_{\cQ}}(t_{\cP},t_{\cQ},1/3)
        \asymp
        \frac{1}{\varepsilon^2}.
    \end{align*}
    The lower bounds now follow from \Cref{lem:general-bernoulli-embedding}.

    Notice that the likelihood-ratio assumption already implies
    \begin{align*}
        \dtv(\cP,\cQ)
        =
        \frac{1}{2}\sum_{x\in[k]}|\cP(x)-\cQ(x)|
        \leq
        \frac{\varepsilon}{2}.
    \end{align*}
    Hence the additional assumption $\dtv(\cP,\cQ)\asymp\varepsilon$ amounts to requiring the total variation distance to be as large as the likelihood-ratio constraint permits, up to constant factors. For the matching upper bounds, let
    \begin{align*}
        A=\{x\in[k]:\cP(x)\geq\cQ(x)\}.
    \end{align*}
    The deterministic binary channel $B(x)=\mathbbm{1}_{A}(x)$ preserves total variation:
    \begin{align*}
        \dtv(B\cP,B\cQ)
        =
        \dtv(\cP,\cQ).
    \end{align*}
    Therefore, by \Cref{fact:tv-hel-ineq},
    \begin{align*}
        \hel(B\cP,B\cQ)
        \gtrsim
        \varepsilon^2.
    \end{align*}
    Under either input distribution, one use of $B$ reveals at most one bit of information. Applying \Cref{thm:general-channel-characterisation} with this channel gives the matching $O(1/\varepsilon^2)$ upper bounds.
\end{proof}

\begin{remark}\label{rem:bounded-likelihood-rare-event}
    The likelihood-ratio condition alone does not determine the information complexity. To see this, let $\rho\in(0,1/4]$ be small and consider
    \begin{align*}
        \cP=\ber((1+\varepsilon)\rho),
        \qquad
        \cQ=\ber(\rho).
    \end{align*}
    The likelihood ratio equals $1+\varepsilon$ at $1$ and equals $1-\varepsilon\rho/(1-\rho)$ at $0$, so it lies in $[1-\varepsilon,1+\varepsilon]$. However, \Cref{thm:coin-info-section3} gives, under either hypothesis,
    \begin{align*}
        \icht_{\cR}(\cP,\cQ,1/3)
        \asymp
        \frac{\log(1/\rho)}{\varepsilon^2}.
    \end{align*}
    Thus the lower bound in \Cref{thm:bounded-likelihood-ratio-lower-bound} can be loose. In this example $\dtv(\cP,\cQ)=\varepsilon\rho\ll\varepsilon$, so it lies outside the sharpness regime of the theorem.
\end{remark}

\begin{remark}
    Bounded-likelihood lower bounds of order $1/\varepsilon^2$ can also be recovered through strong data-processing inequalities~\cite{Duchi,braverman2016communication}. These works do not, however, give a matching information-complexity characterisation for arbitrary $\cP$ and $\cQ$. In particular, although Braverman et al.~\cite{braverman2016communication} show that their distributed strong data-processing inequality is tight for some protocols, these protocols may not solve the testing problem (i.e., induce a constant separation between the transcripts under each hypothesis). The Bernoulli embedding above gives a direct proof from the coin-problem characterisation, while the deterministic binary channel in the matching upper bound identifies $\dtv(\cP,\cQ)\asymp\varepsilon$ as a sufficient condition for sharpness.
\end{remark}

\subsection{Upper bounds via $\chi^2$ divergence}\label{subsec:chisq-upper-bounds}

We finish the section with a general upper bound obtained from a noisy binary channel. 

\begin{theorem}[$\chi^2$ upper bounds]\label{thm:chisq-upper-bounds}
    Let $\cP$ and $\cQ$ be distinct distributions on $[k]$. For every $\delta\in[1/3,1)$,
    \begin{align*}
        \icht_{\cP}(\cP,\cQ,\delta)
        &\lesssim
        \frac{\log(1/(1-\delta))}{\chisq(\cQ,\cP)},
        \\
        \icht_{\cQ}(\cP,\cQ,\delta)
        &\lesssim
        \frac{\log(1/(1-\delta))}{\chisq(\cP,\cQ)},
    \end{align*}
    with the convention that $1/\infty=0$. In particular, at advantage $1/3$ the logarithmic factor is a universal constant.
\end{theorem}

\begin{proof}
    We prove the bound under $\cP$; the other bound follows by interchanging $\cP$ and $\cQ$. First suppose that $\cQ\not\ll\cP$. There is then a set $A\subseteq[k]$ such that $\cP(A)=0$ and $\cQ(A)>0$. The deterministic channel $B(x)=\mathbbm{1}_A(x)$ has zero information cost under $\cP$, while finitely many independent uses of $B$ distinguish the two hypotheses with any advantage smaller than one. Hence $\icht_{\cP}(\cP,\cQ,\delta)=0$, in agreement with $\chisq(\cQ,\cP)=\infty$.

    Now suppose that $\cQ\ll\cP$. Fix a function $f:[k]\to[-1,1]$ with $\var_{\cP}(f(Z))>0$ and $\E_{\cP}f(Z)\neq\E_{\cQ}f(Z)$, and define a binary channel $B_f$ by
    \begin{align*}
        B_f(1\mid x)
        =
        \frac{1}{2}+\frac{1}{4}f(x).
    \end{align*}
    Write
    \begin{align*}
        \mu_{\cR}=\E_{\cR}f(Z),
        \qquad
        p_{\cR}=\frac{1}{2}+\frac{1}{4}\mu_{\cR},
        \qquad
        \cR\in\{\cP,\cQ\}.
    \end{align*}
    Thus $B_f\cR=\ber(p_{\cR})$, and all the Bernoulli parameters involved lie in $[1/4,3/4]$. The standard quadratic bounds for binary relative entropy and Hellinger divergence therefore give
    \begin{align*}
        \mi_{\cP}(Z;B_f(Z))
        &=
        \E_{\cP}\!\left[
            \dkl\left(
                \ber\left(\frac{1}{2}+\frac{1}{4}f(Z)\right)
                \,\middle\|\,
                \ber(p_{\cP})
            \right)
        \right]
        \lesssim
        \var_{\cP}(f(Z)),
        \\
        \hel(B_f\cP,B_f\cQ)
        &\asymp
        \left(\mu_{\cP}-\mu_{\cQ}\right)^2.
    \end{align*}
    Let $n$ be the sample complexity for testing $B_f\cP$ against $B_f\cQ$ with advantage $\delta$. By \Cref{fact:hyp_sample_comp_fact}, the non-interactive protocol in which every agent independently applies $B_f$ and broadcasts its output may take
    \begin{align*}
        n
        \lesssim
        \frac{\log(1/(1-\delta))}
        {\hel(B_f\cP,B_f\cQ)}.
    \end{align*}
    Independence and the chain rule then give
    \begin{align*}
        \icht_{\cP}(\cP,\cQ,\delta)
        &\leq
        n\mi_{\cP}(Z;B_f(Z))
        \\
        &\lesssim
        \frac{\var_{\cP}(f(Z))}
        {\left(\E_{\cP}f(Z)-\E_{\cQ}f(Z)\right)^2}
        \log\frac{1}{1-\delta}.
    \end{align*}
    Taking the infimum over $f$ and applying the variational characterisation in \Cref{fact:chisq_var}, with the roles of $\cP$ and $\cQ$ interchanged, proves
    \begin{align*}
        \icht_{\cP}(\cP,\cQ,\delta)
        \lesssim
        \frac{\log(1/(1-\delta))}{\chisq(\cQ,\cP)}.
    \end{align*}
\end{proof}

\begin{remark}
    When $\cQ\ll\cP$, the variational optimiser is, up to an affine rescaling, the likelihood ratio $d\cQ/d\cP$. Thus the protocol sends a noisy bit whose bias is an affine function of the likelihood ratio. The reversal in \Cref{thm:chisq-upper-bounds} comes from measuring the variance of this statistic.
\end{remark}

\begin{remark}
    The $\chi^2$ upper bound is not always sharp. For the example in \Cref{rem:bounded-likelihood-rare-event},
    \begin{align*}
        \cP=\ber((1+\varepsilon)\rho),
        \qquad
        \cQ=\ber(\rho),
    \end{align*}
    both directed $\chi^2$ divergences are of order $\varepsilon^2\rho$. Hence \Cref{thm:chisq-upper-bounds} gives an upper bound of order $1/(\varepsilon^2\rho)$ at constant advantage, whereas the information complexity under either hypothesis is of order $\log(1/\rho)/\varepsilon^2$.

    On the other hand, under the bounded-likelihood assumption of \Cref{thm:bounded-likelihood-ratio-lower-bound},
    \begin{align*}
        4\dtv^2(\cP,\cQ)
        &\leq
        \chisq(\cP,\cQ)
        \leq
        \varepsilon^2,
        \\
        4\dtv^2(\cP,\cQ)
        &\leq
        \chisq(\cQ,\cP)
        \leq
        \frac{\varepsilon^2}{1-\varepsilon}.
    \end{align*}
    Consequently, the additional condition $\dtv(\cP,\cQ)\asymp\varepsilon$ implies
    \begin{align*}
        \chisq(\cP,\cQ)
        \asymp
        \chisq(\cQ,\cP)
        \asymp
        \varepsilon^2.
    \end{align*}
    In this regime, \Cref{thm:chisq-upper-bounds} matches the lower bound in \Cref{thm:bounded-likelihood-ratio-lower-bound} under both hypotheses.
\end{remark}

\section{Applications}\label{sec:applications}

We give two applications of our information-complexity bounds for the coin problem. First, we use the bound for testing highly biased Bernoulli distributions, together with an information-complexity direct-sum argument, to recover the broadcast-model communication lower bound for set disjointness due to \cite{braverman2015information}. Second, we use the broadcast-to-streaming framework of \cite{streaming2025} to strengthen multi-pass streaming lower bounds. Besides improving the primitive Bernoulli-testing bound in the low-bias regime, we obtain an additional $\log(1/\nu)$ factor for the sparse-background analogue of the core task of \cite{brown2022strong,streaming2025}.

\subsection{Set disjointness in the broadcast model}\label{subsec:braverman_setdisj}

In the $k$-party set-disjointness problem, player $i$ receives a set $X_i\subseteq[n]$, and the players must determine whether their sets have a common element. Identifying each set with its indicator vector, the function to be computed is
\begin{align*}
    \bigvee_{j\in[n]}\bigwedge_{i\in[k]}X_i(j).
\end{align*}
Thus, set disjointness consists of $n$ copies of the $k$-party \textsc{And} function. Braverman and Oshman \cite{braverman2015information} proved an $\Omega(n\log k)$ communication lower bound in the broadcast model by combining a direct-sum argument with a conditional information-complexity lower bound for a single copy of \textsc{And}. We show that our coin-problem theorem gives a simpler route to the latter bound.

The direct-sum argument uses a distribution supported on the $0$-outputs of \textsc{And}. It also requires the players' inputs to be independent after conditioning on an auxiliary random variable. To construct the hard distribution for \textsc{And} on $X_1,\dots,X_k$, \cite{braverman2015information} chooses $Z$ uniformly from $[k]$, sets $X_Z=0$, and draws the remaining coordinates independently from $\ber(1-1/k)$. Their main technical step shows that every protocol computing \textsc{And} with constant advantage satisfies
\begin{align*}
    \mi(X;\Pi\mid Z)\gtrsim\log k.
\end{align*}
The difficulty is not obtaining a coin-problem lower bound at bias $1-1/k$, but transferring such a bound to this conditional distribution. Indeed, an i.i.d.\ vector with coordinates distributed as $\ber(1-1/k)$ has only one zero in expectation. Forcing the coordinate $X_Z$ to be zero therefore changes the input distribution substantially, and there is no high-probability collection of naturally occurring zeros over which one can average. Braverman and Oshman overcome this obstacle through a delicate analysis of the posterior distribution of $Z$ given the transcript.

We use a slightly different distribution. We again choose $Z$ uniformly from $[k]$ and set $X_Z=0$, but now draw the remaining coordinates independently from $\ber(1-1/\sqrt{k})$. An i.i.d.\ vector from this background distribution has $\Theta(\sqrt{k})$ zeros with high probability. The forced zero can therefore be averaged over a large collection of natural zeros, which is what makes the transfer from the coin problem elementary. This choice allows us to invoke directly the information-complexity lower bound for testing $\ber(1)$ against $\ber(1-1/\sqrt{k})$, or equivalently $\ber(0)$ against $\ber(1/\sqrt{k})$ after complementing the bits. Under the non-degenerate hypothesis, this coin problem has information complexity $\Theta(\log k)$ by \Cref{thm:coin-info-section3}. The following lemma transfers that lower bound to the required conditional information cost.

\begin{restatable}{lemma}{bravermandisj}\label{lem:brav_disj}
    Let $Z$ be uniformly distributed on $[k]$. Conditional on $Z$, set $X_Z=0$ and draw $X_i\sim\ber(1-1/\sqrt{k})$ independently for $i\neq Z$. Let $\Pi$ be a broadcast protocol that computes $\bigwedge_{i\in[k]}b_i$ with advantage $\delta\in(0,1)$ for every $b\in\{0,1\}^k$. Then, for all sufficiently large $k$,
    \begin{align*}
        \mi(X;\Pi(X)\mid Z)
        \gtrsim
        \delta^2\log k.
    \end{align*}
\end{restatable}

The proof, given in \Cref{app:set-disjointness-application}, compares this conditional distribution with an i.i.d.\ vector $Y\sim\ber(1-1/\sqrt{k})^{\otimes k}$. Our coin-problem lower bound gives $\mi(Y;\Pi(Y))=\Omega(\delta^2\log k)$. Since $Y$ contains $\Theta(\sqrt{k})$ zero coordinates with high probability, averaging over the possible locations of a forced zero transfers a constant fraction of this information cost to $\mi(X;\Pi(X)\mid Z)$. This yields \Cref{lem:brav_disj}.

The distribution in \Cref{lem:brav_disj} is supported on the $0$-outputs of \textsc{And}, and its coordinates are independent conditional on $Z$. It therefore satisfies the two properties required by the information-complexity direct-sum theorem used in \cite{bar2004information,braverman2015information}. Applying that theorem across the $n$ coordinates gives an $\Omega(n\log k)$ information lower bound for set disjointness. Since information cost is a lower bound on communication, this recovers the corresponding broadcast-model communication lower bound of Braverman and Oshman.
\subsection{Multi-pass streaming information complexity}\label{subsec:streaming-information-application}

We next apply our bound for testing $\ber(0)$ against $\ber(\alpha)$ to the multi-pass streaming information-complexity framework of \cite{streaming2025}. Let $\cA$ be an $L$-pass streaming algorithm operating on a stream $X=(X_1,\ldots,X_N)$, and let $M_{i,r}$ be its memory state after processing $X_i$ during pass $r$. Following \cite{braverman2024newinfocomp,streaming2025}, set $M_{0,r}=M_{N,r-1}$ and write $M_{i,\leq r}=(M_{i,1},\ldots,M_{i,r})$. The multi-pass streaming information complexity of $\cA$ under $X\sim\ber(\theta)^{\otimes N}$ is
\begin{align*}
    \ics_\theta(\cA,X)
    := {}&
    \sum_{r\in[L]}\sum_{i\in[N]}\sum_{k=1}^{i}
    \mi_\theta\lp
        M_{i,r};X_k
        \mid M_{k-1,\leq r},M_{i,\leq r-1}
    \rp
    \\
    &+
    \sum_{r\in[L]}\sum_{i\in[N]}\sum_{k=i+1}^{N}
    \mi_\theta\lp
        M_{i,r};X_k
        \mid M_{k-1,\leq r-1},M_{i,\leq r-1}
    \rp.
\end{align*}
Unlike ordinary transcript information, this quantity can charge information about one input coordinate at several later memory states. It is nevertheless controlled by the memory usage: if $\cA$ uses at most $S$ bits of memory, then \cite[Lemma~10]{streaming2025} gives
\begin{align}\label{eq:streaming-ic-space-upper}
    \ics_\theta(\cA,X)\leq 2LSN.
\end{align}

We first record the form of the lower bound from \cite{streaming2025} that is supported by their proof. 

\begin{theorem}[Variant of Theorem~9 in \cite{streaming2025}]\label{thm:lzw-streaming-strengthened}
    Let $0\leq\beta<\alpha\leq1$, and suppose that an $L$-pass streaming algorithm $\cA$ distinguishes $\ber(\beta)^{\otimes N}$ from $\ber(\alpha)^{\otimes N}$ with advantage $\delta\in(0,1)$. Then
    \begin{align*}
        \ics_\alpha(\cA,X)
        &\gtrsim
        \min\{1/2,\alb\}
        \frac{\alpha^2}{(\alpha-\beta)^2}\delta^3,
        &&X\sim\ber(\alpha)^{\otimes N},
        \\
        \ics_\beta(\cA,X)
        &\gtrsim
        \min\{1/2,\beta\}
        \frac{\betb^2}{(\alpha-\beta)^2}\delta^3,
        &&X\sim\ber(\beta)^{\otimes N}.
    \end{align*}
\end{theorem}

The proof in \cite{streaming2025} represents a sample from $\ber(\beta)$ by first drawing a random subset $S\subseteq[N]$, including each coordinate independently with probability
\begin{align*}
    \lambda=\frac{\alpha-\beta}{\alpha},
\end{align*}
and then forcing the coordinates in $S$ to zero while drawing the remaining coordinates independently from $\ber(\alpha)$. For each fixed $S$, their block simulation embeds a broadcast instance of testing $\ber(0)$ against $\ber(\alpha)$. A local information term appears only when two relevant locations belong to $S$, which occurs with probability of order $\lambda^2$. Consequently, a broadcast lower bound $B_\alpha(0,\alpha,\delta)$ lifts to a streaming lower bound with an extra factor $1/\lambda^2$. Applying the bound in \Cref{thm:colt25}, and repeating the argument after complementing the input bits, gives \Cref{thm:lzw-streaming-strengthened}.

\begin{remark}
    The statement of Theorem~9 in \cite{streaming2025} differs from \Cref{thm:lzw-streaming-strengthened} in two respects. First, its two prefactors are $\min\{\alpha,\alb\}$ and $\min\{\beta,\betb\}$, respectively. In the low-bias directions, the argument can instead use the strengthened endpoint bound in \Cref{thm:colt25}, giving $\min\{1/2,\alb\}$ and $\min\{1/2,\beta\}$. Second, the theorem is stated with a factor $\delta^2$, but the proof actually yields a factor of $\delta^3$. This does not affect the constant-advantage setting.
\end{remark}

For the low-bias regime, the domination of full-transcript information by multi-pass streaming information complexity allows us to apply \Cref{thm:coin-info-section3} directly. This gives the following logarithmic improvement over \Cref{thm:lzw-streaming-strengthened}.

\begin{theorem}[Low-bias multi-pass streaming lower bound]\label{thm:streaming-regime-one-improvement}
    Let $0\leq\beta<\alpha\leq1/2$, and suppose that an $L$-pass streaming algorithm $\cA$ distinguishes $\ber(\beta)^{\otimes N}$ from $\ber(\alpha)^{\otimes N}$ with advantage $\delta\in(0,1)$. Then, for $X\sim\ber(\alpha)^{\otimes N}$,
    \begin{align*}
        \ics_\alpha(\cA,X)
        \gtrsim
        \frac{\ent(\alpha)}{\hel(\alpha,\beta)}\delta^2
        \asymp
        \frac{\alpha^2\log(1/\alpha)}{(\alpha-\beta)^2}\delta^2.
    \end{align*}
\end{theorem}

\begin{proof}
\cite[Claim~11]{streaming2025} gives that the streaming information complexity is at least as large as the blackboard information complexity, and hence
    \begin{align*}
        \ics_\alpha(\cA,X) &\geq
        \icht_\alpha(\delta) \gtrsim \frac{\ent(\alpha)}{\hel(\alpha,\beta)}\delta^2.
    \end{align*}
Note the extra $\log(1/\alpha)$ factor compared to the bound for this regime in \cite{streaming2025}, and the improvement of $\delta^3$ to $\delta^2$.
\end{proof}

\begin{remark}
    The same blackboard simulation applies in the other two regimes, but it does not give a comparable new parameter dependence. If $\beta<1/2<\alpha$, it gives
    \begin{align*}
        \ics_\alpha(\cA,X)
        \gtrsim
        \frac{\alb}{\hel(\alpha,\beta)}\delta^2
        \asymp
        \frac{\alb}{(\alpha-\beta)^2}\delta^2,
    \end{align*}
    which has the same parameter dependence as \Cref{thm:lzw-streaming-strengthened}; only the dependence on $\delta$ improves. If $1/2\leq\beta<\alpha$, the direct simulation instead gives
    \begin{align*}
        \ics_\alpha(\cA,X)
        \gtrsim
        \frac{\alb}{\betb}
        \frac{\ent(\beta)}{\hel(\alpha,\beta)}\delta^2
        \asymp
        \frac{\alb\betb\log(1/\betb)}{(\alpha-\beta)^2}\delta^2.
    \end{align*}
    For constant advantage, this is weaker than the bound $\ics_\alpha(\cA,X)\gtrsim\alb/(\alpha-\beta)^2$ from \Cref{thm:lzw-streaming-strengthened} by a factor of order $\betb\log(1/\betb)$. 
\end{remark}

For constant advantage, \Cref{thm:streaming-regime-one-improvement} strengthens \Cref{thm:lzw-streaming-strengthened} by a factor of order $\log(1/\alpha)$. Combining the theorem with \Cref{eq:streaming-ic-space-upper} also gives
\begin{align*}
    S
    \gtrsim
    \frac{\alpha^2\log(1/\alpha)}{LN(\alpha-\beta)^2}\delta^2.
\end{align*}

\subsubsection{Core task in \cite{brown2022strong}}\label{subsec:sparse-background-core}

We now give a direct application to the core task of Brown, Bun, and Smith \cite{brown2022strong}, using the multi-pass reductions of Li, Wang, and Zhang \cite{streaming2025}. Their construction uses the uniform background $\ber(1/2)^{\otimes d}$ and fixes each hidden coordinate across the stream to an independent $\ber(1/2)$ value. We consider the biased analogue: the background is $\ber(\nu)^{\otimes d}$, and each hidden value is drawn once from $\ber(\nu)$. Thus, after averaging over the hidden pattern, every individual feature vector still has the background law, but it has only about $\nu d$ nonzero coordinates when $\nu$ is small.

We first note the coin problem of testing $\ber(0)$ vs $\ber(\nu)$ that emerges as the relevant one-dimensional problem from this modification.

\begin{lemma}\label{lem:repeated-biased-bit}
    Let $0<\nu\leq1/2$, and let $\Pi$ be a protocol on $T$ input bits. Under the structured hypothesis, draw $B\sim\ber(\nu)$ and set $X=(B,\ldots,B)$. Under the background hypothesis, draw $X\sim\ber(\nu)^{\otimes T}$. If $\Pi$ distinguishes these hypotheses with advantage at least $a\in(0,1)$, then, under the background hypothesis,
    \begin{align*}
        \mi_\nu(X;\Pi)
        \gtrsim
        a^2\log(1/\nu).
    \end{align*}
\end{lemma}

\begin{proof}
    Let $\Pi_0$, $\Pi_1$, and $\Pi_\nu$ denote the transcript laws when the input is respectively the all-zero stream, the all-one stream, and an i.i.d.\ $\ber(\nu)$ stream. The structured transcript law is
    \begin{align*}
        \Pi_{\mathrm{rep},\nu}
        =
        (1-\nu)\Pi_0+\nu\Pi_1.
    \end{align*}
    Set $a_b=\dtv(\Pi_b,\Pi_\nu)$ for $b\in\{0,1\}$. Convexity of total variation gives
    \begin{align*}
        a
        \leq
        (1-\nu)a_0+\nu a_1.
    \end{align*}
    Hence either $(1-\nu)a_0\geq a/2$ or $\nu a_1\geq a/2$. In the first case, applying \Cref{thm:coin-info-section3} to $\ber(0)$ versus $\ber(\nu)$, with information measured under $\ber(\nu)$, yields
    \begin{align*}
        \mi_\nu(X;\Pi)
        \gtrsim
        a_0^2\log(1/\nu)
        \gtrsim
        a^2\log(1/\nu).
    \end{align*}
    In the second case, the same theorem applied to $\ber(1)$ versus $\ber(\nu)$ gives
    \begin{align*}
        \mi_\nu(X;\Pi)
        \gtrsim
        \nu a_1^2
        \gtrsim
        \frac{a^2}{\nu}
        \gtrsim
        a^2\log(1/\nu),
    \end{align*}
    where the last inequality uses $\nu\leq1/2$.
\end{proof}

We next define the biased core task by replacing $\ber(1/2)$ by $\ber(\nu)$ in the description of the same task in \cite{brown2022strong}. For $0<\nu\leq1/2$ and integers $d$ and $1\leq\rho<d$, draw $R$ uniformly from $\{0,\ldots,\rho\}$, draw a uniform $R$-subset $J\subseteq[d]$, and draw $(B_i)_{i\in J}$ independently from $\ber(\nu)$. Let $P^{(\nu)}_{J,B}$ be the distribution on $\{0,1\}^d$ that fixes $X_i=B_i$ for $i\in J$ and draws the remaining coordinates independently from $\ber(\nu)$. Notice that, conditional on $J$ but after averaging over $B$,
\begin{align*}
    X\sim\ber(\nu)^{\otimes d},
    \qquad
    \E\|X\|_0=\nu d.
\end{align*}

Independently draw $k$ such distributions $P^{(\nu)}_{J_1,B_1},\ldots,P^{(\nu)}_{J_k,B_k}$, and let $P^{(\nu)}_{\mathrm{mix}}$ be their uniform labelled mixture: a sample is obtained by drawing $j$ uniformly from $[k]$, drawing $Y\sim P^{(\nu)}_{J_j,B_j}$, and returning $(j,Y)$. In the biased core task, an algorithm receives $N$ i.i.d.\ samples from $P^{(\nu)}_{\mathrm{mix}}$ and outputs a possibly random function $m:[k]\times\{0,1\}^d\to\{0,1\}$. Its advantage is
\begin{align*}
    \E\left[
        \Pr_{(j,Y)\sim P^{(\nu)}_{\mathrm{mix}}}[m(j,Y)=1]
        -
        \Pr_{\substack{j\sim\operatorname{Unif}([k])\\Y\sim\ber(\nu)^{\otimes d}}}[m(j,Y)=1]
    \right],
\end{align*}
where the expectation is over the hidden parameters, the training stream, and the randomness of the algorithm.

Write $\boldsymbol J=(J_1,\ldots,J_k)$ and $\boldsymbol B=(B_1,\ldots,B_k)$. For an auxiliary random variable $U$, the notation $\ics(\cA,X\mid U)$ denotes the quantity obtained from the definition of $\ics(\cA,X)$ by adding $U$ to the conditioning in every mutual-information term.

\begin{theorem}[Core task with sparse background]\label{thm:sparse-background-core}
    There is a universal constant $C>0$ such that the following holds. Let $0<\nu\leq1/2$ and $1\leq\rho\leq d/2$, and suppose that a randomised multi-pass streaming algorithm $\cA$ solves the biased core task with advantage at least a positive universal constant. If
    \begin{align*}
        d\log(1/\nu)
        \geq
        C(\rho+1)^2,
    \end{align*}
    then, for the structured input stream $X$,
    \begin{align*}
        \ics\lp
            \cA,X
            \mid \boldsymbol J,\boldsymbol B
        \rp
        \gtrsim
        \frac{k^2d\log(1/\nu)}{(\rho+1)^2}.
    \end{align*}
\end{theorem}

\begin{remark}
    Consider, for example, $\nu=\log(d)/d$. Each individual feature vector then has marginal distribution $\ber(\nu)^{\otimes d}$, and hence
    \begin{align*}
        \E\|X\|_0
        =
        \log d.
    \end{align*}
    Standard concentration bounds imply that $\|X\|_0=\Theta(\log d)$ with high probability. Moreover, $\log(1/\nu) = \Theta(\log d).$
    Thus, provided $\rho\lesssim\sqrt{d\log d}$, \Cref{thm:sparse-background-core} gives
    \begin{align*}
        \ics\lp
            \cA,X
            \mid \boldsymbol J,\boldsymbol B
        \rp
        \gtrsim
        \frac{k^2d\log d}{(\rho+1)^2}.
    \end{align*}
    This is an additional factor of order $\log d$ over the corresponding uniform-background bound, while each individual feature vector has only $\Theta(\log d)$ nonzero coordinates with high probability.
\end{remark}

\begin{proof}
    We follow the proof of \cite[Theorem~30]{streaming2025}. Its reduction from the core task to Task B$'$ uses \cite[Theorem~28]{streaming2025}; the latter in turn uses the Task B bound in \cite[Theorem~27]{streaming2025}, whose proof proceeds through \cite[Claims~25 and~26]{streaming2025}. These reductions use only the product structure of the background and the fact that the planted coordinate has the same marginal as the background. They therefore remain valid when the $\ber(1/2)$ is replaced by $\ber(\nu)$ in the structured and background constructions.

    The only part that changes is the Task A information bound in \cite[Lemma~18]{streaming2025}. Replacing that lemma by the bound in \Cref{lem:repeated-biased-bit} leads to the extra $\log(1/\nu)$ factor throughout the proof. We skip additional details here to avoid repetition.
\end{proof}
\section{Conclusion}\label{sec:conclusion}

We have characterised, up to universal constant factors, the information complexity of the coin problem in the broadcast model under either hypothesis and for every pair of Bernoulli parameters in the constant-advantage regime. The characterisation exhibits three different regimes where different protocols are optimal. The main technical result used to prove the lower bound is a novel mixed Hellinger--Jensen--Shannon inequality. We also considered simple binary hypothesis testing between arbitrary discrete distributions. We characterised its constant-advantage information complexity by an optimisation over channels and showed that binary-output channels suffice up to a universal constant. We also obtained interpretable lower and upper bounds in some cases. Finally, we applied the coin-problem bounds to recover the broadcast-model set-disjointness lower bound and strengthened some existing multi-pass streaming lower bounds.


An open question is characterising the tight streaming information complexity of the coin problem, either single-pass or multi-pass. Consider the single-pass setting for now. For testing $\ber(0)$ versus $\ber(\alpha)$ for small $\alpha$, \Cref{thm:streaming-regime-one-improvement} implies a lower bound of $\log(1/\alpha)$ for any single-pass streaming algorithm.
However, the most obvious algorithm maintains $M_t=\bigvee_{j=1}^{t}X_j$ for $t\leq\lceil C/\alpha\rceil$ and outputs $M_{\lceil C/\alpha\rceil}$; for a suitable constant $C$, this has constant advantage.
The streaming information complexity of this algorithm is $\asymp 1/\alpha$, which is probably the correct dependence.

It would also be interesting to understand whether the mixed Hellinger--Jensen--Shannon inequality has applications beyond the coin problem. This inequality may be useful for proving information-complexity lower bounds for other distributed hypothesis-testing and communication problems in which standard Hellinger or Jensen--Shannon inequalities are not strong enough.

\paragraph{Statement on AI use} The key technical result in  \Cref{thm:mixed-hellinger-js} was proved by hand around January 2026, at which point the available models were not strong enough to prove it. We believe current models would likely be successful if pointed in the right direction. (It is hard to be certain since our results are probably part of OpenAI's training data at this point). We used OpenAI's models for the following parts of the paper: \Cref{lem:psdpi_ber} (computing the post-SDPI constant), \Cref{lem:brav_disj} (proving the conditional information complexity lower-bound for set disjointness communication lower bound) and \Cref{thm:binary-output-suffices} (reduction to binary output channels). These were proved by AI once given the correct statement. AI was also used to prepare the manuscript.

\paragraph{Acknowledgements} Both authors gratefully acknowledge support from the Leverhulme Trust through the grant RPG-2025-226 titled ``Old Problems, New
Perspectives: A Fresh Look at Classical Hypothesis Testing.'' 

\bibliography{refrences}

@inproceedings{braverman2016communication,
  title={Communication lower bounds for statistical estimation problems via a distributed data processing inequality},
  author={M. Braverman and A. Garg and T. Ma and H. L. Nguyen and D. P. Woodruff},
  booktitle={Proceedings of the Forty-Eighth Annual {ACM} Symposium on Theory of Computing},
  pages={1011--1020},
  year={2016}
}

@article{bar2004information,
  title={An information statistics approach to data stream and communication complexity},
  author={Z. Bar-Yossef and T. S. Jayram and R. Kumar and D. Sivakumar},
  journal={Journal of Computer and System Sciences},
  volume={68},
  number={4},
  pages={702--732},
  year={2004},
  publisher={Elsevier}
}

@inproceedings{jayram2009hellinger,
  title={{Hellinger} strikes back: {A} note on the multi-party information complexity of {AND}},
  author={T. S. Jayram},
  booktitle={International Workshop on Approximation Algorithms for Combinatorial Optimization},
  pages={562--573},
  year={2009},
  organization={Springer}
}

@article{streaming2025,
  title={Multi-Pass Memory Lower Bounds for Learning Problems},
  author={Q. Li and S. Wang and J. Zhang},
  journal={Proceedings of Machine Learning Research},
  volume={291},
  pages={3671--3699},
  year={2025}
}

@article{harremoes2011pairs,
  title={On pairs of $f$-divergences and their joint range},
  author={P. Harremo{\"e}s and I. Vajda},
  journal={IEEE Transactions on Information Theory},
  volume={57},
  number={6},
  pages={3230--3235},
  year={2011},
  publisher={IEEE}
}

@book{powu2025,
  title={Information theory: From coding to learning},
  author={Y. Polyanskiy and Y. Wu},
  year={2025},
  publisher={Cambridge University Press}
}

@inproceedings{hadar2019communication,
  title={Communication complexity of estimating correlations},
  author={U. Hadar and J. Liu and Y. Polyanskiy and O. Shayevitz},
  booktitle={Proceedings of the 51st Annual {ACM} {SIGACT} Symposium on Theory of Computing},
  pages={792--803},
  year={2019}
}

@inproceedings{braverman2015information,
  title={On information complexity in the broadcast model},
  author={M. Braverman and R. Oshman},
  booktitle={Proceedings of the 2015 {ACM} Symposium on Principles of Distributed Computing},
  pages={355--364},
  year={2015}
}

@inproceedings{bravermanICM,
  title={Communication and information complexity},
  author={M. Braverman},
  booktitle={International Congress of Mathematicians},
  pages={284--320},
  year={2023},
  organization={EMS Press}
}

@inproceedings{braverman2013tight,
  title={A tight bound for set disjointness in the message-passing model},
  author={M. Braverman and F. Ellen and R. Oshman and T. Pitassi and V. Vaikuntanathan},
  booktitle={2013 {IEEE} 54th Annual Symposium on Foundations of Computer Science},
  pages={668--677},
  year={2013},
  organization={IEEE}
}

@inproceedings{woodruff2012tight,
  title={Tight bounds for distributed functional monitoring},
  author={D. P. Woodruff and Q. Zhang},
  booktitle={Proceedings of the Forty-Fourth Annual {ACM} Symposium on Theory of Computing},
  pages={941--960},
  year={2012}
}

@inproceedings{brown2022strong,
  title={Strong memory lower bounds for learning natural models},
  author={G. Brown and M. Bun and A. Smith},
  booktitle={Conference on Learning Theory},
  pages={4989--5029},
  year={2022},
  organization={PMLR}
}

@inproceedings{dagan2018detecting,
  title={Detecting correlations with little memory and communication},
  author={Y. Dagan and O. Shamir},
  booktitle={Conference on Learning Theory},
  pages={1145--1198},
  year={2018},
  organization={PMLR}
}

@article{garg2026unified,
  title={A Unified Approach to Memory-Sample Tradeoffs for Detecting Planted Structures},
  author={S. Garg and J. Hastings and C. Pabbaraju and V. Sharan},
  journal={{arXiv} preprint arXiv:2603.00770},
  year={2026}
}

@inproceedings{braverman2020coin,
  title={The coin problem with applications to data streams},
  author={M. Braverman and S. Garg and D. P. Woodruff},
  booktitle={2020 {IEEE} 61st Annual Symposium on Foundations of Computer Science ({FOCS})},
  pages={318--329},
  year={2020},
  organization={IEEE}
}

@inproceedings{braverman2024newinfocomp,
  title={A new information complexity measure for multi-pass streaming with applications},
  author={M. Braverman and S. Garg and Q. Li and S. Wang and D. P. Woodruff and J. Zhang},
  booktitle={Proceedings of the 56th Annual {ACM} Symposium on Theory of Computing},
  pages={1781--1792},
  year={2024}
}

@inproceedings{steinberger2013distinguishability,
  title={The distinguishability of product distributions by read-once branching programs},
  author={J. Steinberger},
  booktitle={2013 {IEEE} Conference on Computational Complexity},
  pages={248--254},
  year={2013},
  organization={IEEE}
}

@inproceedings{brody2010coin,
  title={The coin problem and pseudorandomness for branching programs},
  author={J. Brody and E. Verbin},
  booktitle={2010 {IEEE} 51st Annual Symposium on Foundations of Computer Science},
  pages={30--39},
  year={2010},
  organization={IEEE}
}

@inproceedings{braverman2022tight,
  title={Tight space complexity of the coin problem},
  author={M. Braverman and S. Garg and O. Zamir},
  booktitle={2021 {IEEE} 62nd Annual Symposium on Foundations of Computer Science ({FOCS})},
  pages={1068--1079},
  year={2022},
  organization={IEEE}
}

@inproceedings{Duchi,
 author = {Y. Zhang and J. Duchi and M. Jordan and M. J. Wainwright},
 booktitle = {Advances in Neural Information Processing Systems},
 editor = {C. J. Burges and L. Bottou and M. Welling and Z. Ghahramani and K. Weinberger},
 publisher = {Curran Associates, Inc.},
 title = {Information-theoretic lower bounds for distributed statistical estimation with communication constraints},
 url = {https://proceedings.neurips.cc/paper_files/paper/2013/file/d6ef5f7fa914c19931a55bb262ec879c-Paper.pdf},
 volume = {26},
 year = {2013}
}

@article{kearns2013large,
  title={Large deviation methods for approximate probabilistic inference},
  author={M. Kearns and L. Saul},
  journal={{arXiv} preprint arXiv:1301.7392},
  year={2013}
}

@inproceedings{CohEtal14,
  title={Two sides of the coin problem},
  author={G. Cohen and A. Ganor and R. Raz},
  booktitle={Approximation, Randomization, and Combinatorial Optimization. Algorithms and Techniques ({APPROX}/{RANDOM} 2014)},
  pages={618--629},
  year={2014},
  organization={Schloss Dagstuhl--Leibniz-Zentrum f{\"u}r Informatik}
}

@article{LimEtal21,
  title={A Fixed-Depth Size-Hierarchy Theorem for {\(\mathrm{AC}^0[\oplus]\)} via the Coin Problem},
  author={N. Limaye and K. Sreenivasaiah and S. Srinivasan and U. Tripathi and S. Venkitesh},
  journal={SIAM Journal on Computing},
  volume={50},
  number={4},
  pages={1461--1499},
  year={2021},
  publisher={Society for Industrial and Applied Mathematics}
}

@article{LeeVio18,
  title={The coin problem for product tests},
  author={C. H. Lee and E. Viola},
  journal={{ACM} Transactions on Computation Theory ({TOCT})},
  volume={10},
  number={3},
  pages={1--10},
  year={2018},
  publisher={ACM}
}

@article{Cov69,
  title={Hypothesis testing with finite statistics},
  author={T. M. Cover},
  journal={The Annals of Mathematical Statistics},
  volume={40},
  number={3},
  pages={828--835},
  year={1969},
  publisher={Institute of Mathematical Statistics}
}

@article{Kop75,
  title={Necessary and sufficient memory size for m-hypothesis testing},
  author={J. Koplowitz},
  journal={IEEE Transactions on Information Theory},
  volume={21},
  number={1},
  pages={44--46},
  year={1975},
  publisher={IEEE}
}

@article{PenEtal23,
  title={Communication-constrained hypothesis testing: Optimality, robustness, and reverse data processing inequalities},
  author={A. Pensia and V. Jog and P. Loh},
  journal={IEEE Transactions on Information Theory},
  volume={70},
  number={1},
  pages={389--414},
  year={2024},
  publisher={IEEE}
}

@article{PenEtal24,
  title={Simple binary hypothesis testing under local differential privacy and communication constraints},
  author={A. Pensia and A. R. Asadi and V. Jog and P. Loh},
  journal={IEEE Transactions on Information Theory},
  volume={71},
  number={1},
  pages={592--617},
  year={2025},
  publisher={IEEE}
}

@book{LeCam86,
  author    = {L. {Le Cam}},
  title     = {Asymptotic Methods in Statistical Decision Theory},
  series    = {Springer Series in Statistics},
  publisher = {Springer},
  address   = {New York, NY},
  year      = {1986},
  doi       = {10.1007/978-1-4612-4946-7}
}

@inproceedings{PenEtal24b,
  title={The sample complexity of simple binary hypothesis testing},
  author={A. Pensia and V. Jog and P. Loh},
  booktitle={The Thirty Seventh Annual Conference on Learning Theory},
  pages={4205--4206},
  year={2024},
  organization={PMLR}
}

@inproceedings{KazEtal25,
  title={The Sample Complexity of Distributed Simple Binary Hypothesis Testing under Information Constraints},
  author={H. Kazemi and A. Pensia and V. Jog},
  booktitle={The Thirty Eighth Annual Conference on Learning Theory},
  pages={3213--3214},
  year={2025},
  organization={PMLR}
}

@incollection{Tsi93,
  author    = {J. N. Tsitsiklis},
  title     = {Decentralized Detection},
  booktitle = {Advances in Statistical Signal Processing},
  pages     = {297--344},
  publisher = {JAI Press},
  year      = {1993}
}

@article{KOV16,
  author  = {P. Kairouz and S. Oh and P. Viswanath},
  title   = {Extremal Mechanisms for Local Differential Privacy},
  journal = {Journal of Machine Learning Research},
  volume  = {17},
  number  = {17},
  pages   = {1--51},
  year    = {2016},
  url     = {https://jmlr.org/papers/v17/15-135.html}
}

@misc{GNCI26,
  author        = {E. Ghazi and J. Nasser and F. Calmon and I. Issa},
  title         = {Sort, Partition, Randomize: Optimal Binary Hypothesis Testing under Local Differential Privacy},
  year          = {2026},
  eprint        = {2606.07443},
  archivePrefix = {arXiv},
  primaryClass  = {cs.IT},
  doi           = {10.48550/arXiv.2606.07443},
  url           = {https://arxiv.org/abs/2606.07443}
}

@article{shaltiel2010hardness,
  author  = {R. Shaltiel and E. Viola},
  title   = {Hardness Amplification Proofs Require Majority},
  journal = {SIAM Journal on Computing},
  volume  = {39},
  number  = {7},
  pages   = {3122--3154},
  year    = {2010},
  doi     = {10.1137/080735096}
}

@inproceedings{aaronson2010bqp,
  author    = {S. Aaronson},
  title     = {{BQP} and the Polynomial Hierarchy},
  booktitle = {Proceedings of the Forty-Second {ACM} Symposium on Theory of Computing},
  series    = {STOC '10},
  pages     = {141--150},
  publisher = {Association for Computing Machinery},
  year      = {2010},
  doi       = {10.1145/1806689.1806711}
}

\appendix

\section{Auxiliary proofs for the coin problem}\label{app:coin-auxiliary-proofs}

\subsection{Upper-bound channel calculations}\label{app:coin-upper-bound-channel-calculations}

We record the elementary channel estimates used in \Cref{subsec:coin-upper-bounds}. First let $C_\ve$ be the binary symmetric channel with crossover probability $1/2-\ve$. For $\gamma\in[0,1]$,
\begin{align*}
    C_\ve(\gamma)
    =
    \ber\lp1/2+(2\gamma-1)\ve\rp.
\end{align*}
Therefore, whenever $\beta<1/2<\alpha$,
\begin{align*}
    \hel(C_\ve(\alpha),C_\ve(\beta))
    \asymp
    \ve^2(\alpha-\beta)^2.
\end{align*}
Moreover,
\begin{align*}
    \mi_\gamma(X;C_\ve(X))
    &=
    \ent\lp1/2+(2\gamma-1)\ve\rp
    -
    \ent(1/2-\ve).
\end{align*}
To estimate this difference, recall that, for a universal constant $c_0>0$,
\begin{align*}
    \ent(1/2+t)
    =
    \ent(1/2)-c_0t^2+O(t^4).
\end{align*}
Substituting $t=(2\gamma-1)\ve$ and $t=-\ve$, and using the even Taylor expansion of entropy, gives uniformly over $\gamma\in[0,1]$,
\begin{align*}
    \mi_\gamma(X;C_\ve(X))
    &=
    c_0\lp1-(2\gamma-1)^2\rp\ve^2\lp1+O(\ve^2)\rp
    \\
    &=
    4c_0\ve^2\gamma\bar\gamma\lp1+O(\ve^2)\rp
    \asymp
    \ve^2\gamma\bar\gamma.
\end{align*}
Thus, in the straddling regime,
\begin{align*}
    \mi_\alpha(X;C_\ve(X))
    \asymp
    \ve^2\alb.
\end{align*}

Next suppose $1/2\leq\beta<\alpha$ and consider the asymmetric Z-channel used in the third regime,
\begin{align*}
    \cZ_\beta(0)
    \sim
    \ber(0),
    \qquad
    \cZ_\beta(1)
    \sim
    \ber(\beta).
\end{align*}
For an input distributed as $\ber(\gamma)$, its output law is $\cZ_\beta(\gamma)=\ber(\beta\gamma)$. Hence, by \Cref{fact:hel_ber} and its complementary form,
\begin{align*}
    \hel(\cZ_\beta(\alpha),\cZ_\beta(\beta))
    =
    \hel(\alpha\beta,\beta^2)
    \asymp
    \frac{(\alpha-\beta)^2}{\betb}.
\end{align*}
The information revealed under $\ber(\alpha)$ is
\begin{align*}
    \mi_\alpha(X;\cZ_\beta(X))
    =
    \ent(\alpha\beta)-\alpha\ent(\beta).
\end{align*}
Since $\alpha\beta=\beta-\alb\beta$, concavity of entropy and the tangent line at $\beta$ give
\begin{align*}
    \ent(\alpha\beta)
    &\leq
    \ent(\beta)
    +
    \alb\beta\log\lp\frac{\beta}{\betb}\rp.
\end{align*}
It follows that
\begin{align*}
    \mi_\alpha(X;\cZ_\beta(X))
    &\leq
    \alb\left(
        \ent(\beta)
        +
        \beta\log\lp\frac{\beta}{\betb}\rp
    \right)
    \\
    &=
    \alb\log(1/\betb).
\end{align*}

\subsection{Proof of the mixed Hellinger--Jensen--Shannon inequality}\label{app:mixed-hellinger-js}

This subsection proves \Cref{thm:mixed-hellinger-js}. We first record the auxiliary estimates used in the reduction and then assemble them at the end.

\subsubsection{The mixed Hellinger divergence as an $f$-divergence}

\begin{lemma}\label{lem:mixed-hel-f-divergence}
For any $\alpha,\beta\in[0,1]$ and two arbitrary distributions $\cP$ and $\cQ$ on a discrete space $\cX$, $\hel(\alpha \cP+\alb \cQ,\beta \cP+\betb \cQ)$ is an $f$-divergence between $\cP$ and $\cQ$.
\end{lemma}

\begin{proof}[Proof of \Cref{lem:mixed-hel-f-divergence}]
    We first reformulate the expression for $\hel(\alpha \cP+\alb \cQ,\beta \cP+\betb \cQ)$ as follows.
    
    \begin{align*}
        &\hel(\alpha \cP+\alb \cQ,\beta \cP+\betb \cQ)
        =
        \frac{1}{2}\sum_{x\in\cX}\lp\sqrt{\alpha \cP(x)+\alb \cQ(x)}-\sqrt{\beta \cP(x)+\betb \cQ(x)}\rp^2
        =\\&
        \frac{1}{2}\sum_{\substack{x\in\cX s.t.\\ \cQ(x)>0}}\cQ(x)\lp\sqrt{\alpha\frac{d\cP}{d\cQ}(x)+\alb}-\sqrt{\beta\frac{d\cP}{d\cQ}(x)+\betb}\rp^2+\frac{1}{2}(\sqrt{\alpha}-\sqrt{\beta})^2 \cP\lp\{x\in\cX~s.t.~\cQ(x)=0\}\rp
    \end{align*}
    This matches the definition of $f$-divergences with $h_m(x)=1/2\left(\sqrt{\alpha x+\alb}-\sqrt{\beta x+\betb}\right)^2$. Also note that $h_m(1)=0$, so we only need to show that this function is convex. To show the convexity, we show that $S(x,y)=\sqrt{xy}$ is concave on $\reals_{+}^2$. This implies the convexity of $h_m(x)$, since
    \begin{align*}
        h_m(x) 
        =
        1/2\left((\alpha+\beta)x+(\alb +\betb)-2\sqrt{(\alpha x+\alb)(\beta x+\betb)}\right)
    \end{align*}
    Thus, it is enough to prove that $\sqrt{(\alpha x+\alb)(\beta x+\betb)}$ is concave. This function is the composition of $S(x,y)$ with an affine transformation, so it is concave. 
    To prove the concavity of $S(x,y)$, we compute its Hessian matrix and show that it is negative semidefinite.
    \begin{align*}
        \mathrm{H}_S 
        =
        \begin{bmatrix}
            \frac{-\sqrt{y}}{4x^{3/2}}& \frac{1}{4\sqrt{xy}}\\\frac{1}{4\sqrt{xy}}&\frac{-\sqrt{x}}{4y^{3/2}}
        \end{bmatrix}
    \end{align*}
    We have $\det(H_S)= 0$, and the sum of the eigenvalues is $-1/4(\frac{\sqrt{y}}{x^{3/2}}+\frac{\sqrt{x}}{y^{3/2}})\leq 0$, so the non-zero eigenvalue is negative. 
    
    It remains to verify that $\lim_{x\rightarrow\infty} h_m'(x)=\frac{1}{2}(\sqrt{\alpha}-\sqrt{\beta})^2$. We have
    
    \begin{align*}
        \lim_{x\rightarrow\infty}h_m'(x)
        =
        \lim_{x\rightarrow\infty}\frac{1}{2}\left(\alpha+\beta-\frac{2\alpha\beta x+\alpha\betb+\beta\alb}{\sqrt{(\alpha x+\alb)(\beta x+\betb)}}\right)=\frac{1}{2}(\alpha+\beta-2\sqrt{\alpha\beta})
    \end{align*}
\end{proof}

\subsubsection{Reduction to Bernoulli pairs}

By \Cref{lem:mixed-hel-f-divergence}, the map $(\cP,\cQ)\mapsto\hel(\alpha \cP+\alb \cQ,\beta \cP+\betb \cQ)$ is an $f$-divergence between $\cP$ and $\cQ$. Since $\js_\alpha(\cP,\cQ)$ is also an $f$-divergence, the joint-range theorem for pairs of $f$-divergences \cite{harremoes2011pairs} reduces \Cref{thm:mixed-hellinger-js} to the case $\cP=\ber(p)$ and $\cQ=\ber(q)$. In that case the remaining inequality is $f(\alpha,\beta,p,q)\geq0$, where
\begin{equation*}
    f(\alpha,\beta,p,q)
    =
    \js_\alpha(p,q)
    -
    \frac{\alpha^2\log(1/\alpha)}{(\alpha-\beta)^2}
    \hel(\alpha p+\alb q,\beta p+\betb q).
\end{equation*}
The next four auxiliary claims prove this binary inequality.

\subsubsection{Concavity in $\beta$}

\begin{lemma}\label{lem:concavity-in-beta}
For any fixed values $\alpha,p,q\in[0,1]$ and $\beta\neq\alpha$, we have $\frac{\partial^2}{\partial\beta^2}f(\alpha,\beta,p,q)\leq 0$.
\end{lemma}

\begin{proof}[Proof of \Cref{lem:concavity-in-beta}]
    Because $\alpha,p$ and $q$ are fixed, we drop them from the notation of $f$ and write $f(\beta)$ instead. Define $\hel(x,y)=\hel(\ber(x),\ber(y))$, $\ma=\alpha p +\alb q$ and $\mb=\beta p +\betb q$.
    In the non-degenerate case $0<\alpha<1$, set $c_\alpha=\alpha^2\log(1/\alpha)>0$. Then
    \begin{align*}
        f'(\beta)
        &=
        -c_\alpha\left(
        2(\alpha-\beta)^{-3}\hel(\ma,\mb)
        +(\alpha-\beta)^{-2}(p-q)\frac{d}{dy}\hel(\ma,\mb)
        \right),
        \\
        f''(\beta)
        &=
        -c_\alpha\Bigl(
        6(\alpha-\beta)^{-4}\hel(\ma,\mb)
        \notag\\
        &\qquad
        +4(\alpha-\beta)^{-3}(p-q)\frac{d}{dy}\hel(\ma,\mb)
        \notag\\
        &\qquad
        +(\alpha-\beta)^{-2}(p-q)^2\frac{d^2}{dy^2}\hel(\ma,\mb)
        \Bigr).
    \end{align*}

    We prove that $f''(\beta)\leq 0$ for any $\beta\neq\alpha$. Define
    \begin{align*}
        h(\ma,\mb)
        \triangleq
        -\frac{(\alpha-\beta)^4}{6c_\alpha}f''(\beta).
    \end{align*}
         Noting that $(\alpha-\beta)(p-q)=(\ma-\mb)$, we can write $h(\ma,\mb)$ as follows
         \begin{equation*}
             h(\ma,\mb) = \hel(\ma,\mb)+\frac{2}{3}(\ma-\mb)\frac{d}{dy}\hel(\ma,\mb)+\frac{1}{6}(\ma-\mb)^2\frac{d^2}{dy^2}\hel(\ma,\mb)  
         \end{equation*}
         We prove that $h(x,y)$ is convex with respect to $y$, $\frac{d}{dy}h(x,x)=0$ and $h(x,x)=0$, so $h(x,y)\geq 0$ for every $x,y\in[0,1]$. The following calculations verify these claims.
         \begin{align*}
             &\frac{d}{dy}h(x,y)
             = 
             \frac{1}{3}\frac{d}{dy}\hel(x,y)+\frac{1}{3}(x-y)\frac{d^2}{dy^2}\hel(x,y)+\frac{1}{6}(x-y)^2\frac{d^3}{dy^3}\hel(x,y)
             \Rightarrow
             \frac{d}{dy}h(x,x)
             =
             0
             \\
             & \frac{d^2}{dy^2}h(x,y)=\frac{1}{6}(x-y)^2\frac{d^4}{dy^4}\hel(x,y)
             =
             \frac{5}{32}(x-y)^2\lp\sqrt{x}y^{-7/2}+\sqrt{\Bar{x}}\Bar{y}^{-7/2}\rp
             \geq
             0
         \end{align*}
\end{proof}

\subsubsection{The limiting case $\beta=\alpha$}

\begin{lemma}\label{lem:ineq:beta=alpha}
For any $p\neq q\in[0,1]$ and $\alpha\in[0,1/2]$, the limit $\lim_{\beta\rightarrow\alpha}f(\alpha,\beta,p,q)$ exists and is non-negative.
\end{lemma}

\begin{proof}[Proof of \Cref{lem:ineq:beta=alpha}]
    First, we calculate the limit using L'Hôpital's rule. It suffices to calculate the terms depending on $\beta$. We use $\frac{d}{dy}\hel(\ma,\mb)$ to denote $\frac{d}{dy}\hel(x,y)$ evaluated at $x=\ma$ and $y=\mb$, and similarly for the second derivative.
    \begin{align*}
        \lim_{\beta\rightarrow\alpha}\frac{\hel(\ma,\mb)}{(\alpha-\beta)^2}
        =
        \lim_{\beta\rightarrow\alpha}\frac{(p-q)\frac{d}{dy}\hel(\ma,\mb)}{-2(\alpha-\beta)}
        =
        \lim_{\beta\rightarrow\alpha}\frac{(p-q)^2}{2}\frac{d^2}{dy^2}\hel(\ma,\mb)
    \end{align*}
    Therefore, we have 
    \begin{align*}
        \lim_{\beta\rightarrow\alpha}f(\alpha,\beta,p,q)
        =
        \jsa(p,q)-\frac{\alpha^2\log(1/\alpha)}{8}\frac{(p-q)^2}{\ma\mab}
    \end{align*}
It remains to prove that the expression above is non-negative. Recall that $\jsa(p,q)$ is an $f$-divergence between two Bernoulli random variables. First, we prove that the term $\frac{(p-q)^2}{\ma\mab}$ is also an $f$-divergence between $\ber(p)$ and $\ber(q)$, and then we use this fact to prove the result. We have
\begin{align*}
    \frac{(p-q)^2}{\ma\mab} = q\frac{(\frac{p}{q}-1)^2}{\alpha \frac{p}{q} +\alb} + \Bar{q}\frac{(\frac{\Bar{p}}{\Bar{q}}-1)^2}{\alpha \frac{\Bar{p}}{\Bar{q}} +\alb}
\end{align*}
Denote this $f$-divergence by $\mathrm{D}_K(p,q)$ and its defining function by $f_K(x)=\frac{(x-1)^2}{\alpha x+\alb}$. Let $f_{JS}(x)=\alpha x\log(\frac{x}{\alpha x+\alb})+\alb\log(\frac{1}{\alpha x+\alb})$ be the defining function of the Jensen--Shannon divergence. To prove the inequality, it is enough to show that
$f_{JS}(x)\geq(1/8)\alpha^2\log(1/\alpha)f_K(x)$ for every $x\in\real_+\cup\{\infty\}$. After algebraic rearrangement, this is equivalent to
\begin{align*}
    \dkl\lp\frac{\alpha x}{\alpha x+\alb}\|\alpha\rp
    \geq
    \frac{\log(1/\alpha)}{8\alb^2}\lp\frac{\alpha x}{\alpha x+\alb}-\alpha\rp^2
\end{align*}
It is therefore enough to prove that for any $t\in[0,1]$ and $\alpha\in[0,1/2]$ we have 
\begin{equation*}
    \dkl(t\|\alpha)
    \geq
    \frac{\log(1/\alpha)}{2}(t-\alpha)^2
\end{equation*}

We use the following fact, proved in Lemma~1 of \cite{kearns2013large}.
\begin{fact}[\cite{kearns2013large}]
    If $X\sim\ber(\alpha)$ we have
    \begin{equation*}
        \log\lp\E e^{\lambda(X-\alpha)}\rp
        \leq
        \frac{\lambda^2}{4}\frac{1-2\alpha}{\log(\frac{\alb}{\alpha})}
    \end{equation*}
\end{fact}

The variational form of $\dkl$ gives

\begin{align*}
    \dkl(t\|\alpha)=\sup_{\lambda\in\reals}\{\lambda(t-\alpha)-\log(\E_{X\sim\ber(\alpha)}e^{\lambda(X-\alpha)}\}
    \geq
    \frac{\log(\frac{\alb}{\alpha})}{1-2\alpha}(t-\alpha)^2
\end{align*}

It remains to verify that $\frac{\log(\frac{\alb}{\alpha})}{1-2\alpha}\geq \frac{\log(1/\alpha)}{2}$.

\begin{align*}
    \frac{\log(\frac{\alb}{\alpha})}{1-2\alpha}\geq \frac{\log(1/\alpha)}{2}
    \Leftrightarrow
    (1/2+\alpha)\log(1/\alpha)\geq\log(1/\alb)
    \Leftarrow
    (1/2+\alpha)\log(1/\alpha)\geq\log(2)
\end{align*}

The left-hand side is decreasing in $\alpha$ for $\alpha\in[0,1/2]$, and equality holds at $\alpha=1/2$.
\end{proof}

\subsubsection{The boundary case $\beta=0$}

\begin{lemma}\label{lem:ineq:beta=0}
For any $\alpha\in[0,1]$ and arbitrary $p,q\in[0,1]$, we have
\begin{equation*}
    \jsa(p,q)\geq\log(1/\alpha)\hel(\alpha p+\alb q,q).
\end{equation*}
\end{lemma}

\begin{proof}[Proof of \Cref{lem:ineq:beta=0}]
    Note that for $p=q$, the inequality is trivial. For fixed $p\neq q\in[0,1]$, we want to prove the following inequality for any $\alpha\in(0,1)$:
\begin{equation*}
    f(\alpha) = \jsa(p,q)-\log(1/\alpha)\hel(\alpha p+\alb q,q)
    \geq
    0
\end{equation*}
To do so, we show the following two facts:
\begin{enumerate}
    \item $f(1)=\lim_{\alpha\downarrow0}f(\alpha)=0$
    \item $\frac{d^2}{d\alpha^2}f(\alpha)\leq 0$
\end{enumerate}
These facts imply $f(\alpha)\geq0$: extending $f$ continuously to $\alpha=0$ gives a concave function that vanishes at both endpoints and is therefore non-negative on $\alpha\in(0,1)$.

Observe that $f(1)=0$. To calculate $\lim_{\alpha\downarrow0}f(\alpha)$, we use the derivative of $\hel(\alpha p+\alb q,q)$:
\begin{align*}
    \frac{d}{d\alpha}\hel(\alpha p+\alb q,q)
    =
    \frac{(p-q)}{2}\left(\frac{\sqrt{\qb}}{\sqrt{1-\alpha p-\alb q}}-\frac{\sqrt{q}}{\sqrt{\alpha p+\alb q}}\right)
\end{align*}
We next calculate
    $\lim_{\alpha\downarrow0}f(\alpha) = \lim_{\alpha\downarrow0}\log(1/\alpha)\hel(\alpha p+\alb q,q)$. We use the squeeze theorem and L'Hôpital's rule:
\begin{align*}
    0
    \leq
    \log(1/\alpha)\hel(\alpha p+\alb q,q)
    \leq
    (\frac{1}{\alpha}-1)\hel(\alpha p+\alb q,q)
\end{align*}
We know that $\lim_{\alpha\downarrow0}\hel(\alpha p+\alb q,q)=0$ and calculate $\lim_{\alpha\downarrow0}\frac{\hel(\alpha p+\alb q,q)}{\alpha}$ by L'Hôpital's rule:
\begin{align*}
    \lim_{\alpha\downarrow0}\frac{\hel(\alpha p+\alb q,q)}{\alpha}
    =
    \lim_{\alpha\downarrow 0}\frac{(p-q)}{2}\left(\frac{\sqrt{\qb}}{\sqrt{1-\alpha p-\alb q}}-\frac{\sqrt{q}}{\sqrt{\alpha p+\alb q}}\right)
    =
    0
\end{align*}

For the second fact, differentiation gives the second derivative of $f(\alpha)$:

\begin{align}\label{eq:d2f}
    \nonumber
    \frac{d^2}{d\alpha^2}f(\alpha) 
    =&
    \underbrace{\frac{(p-q)}{\alpha}\left(\frac{\sqrt{\qb}}{\sqrt{1-\alpha p-\alb q}}-\frac{\sqrt{q}}{\sqrt{\alpha p+\alb q}}\right)}_{A_1}\underbrace{-\frac{(p-q)^2\log(1/\alpha)}{4}\left(\frac{\sqrt{\qb}}{(1-\alpha p-\alb q)^{3/2}}+\frac{\sqrt{q}}{(\alpha p+\alb q)^{3/2}}\right)}_{A_2}
    \\&
    \underbrace{-(p-q)^2\left(\frac{1}{\alpha p +\alb q}+\frac{1}{1-\alpha p-\alb q}\right)}_{A_3}\underbrace{-\frac{1}{\alpha^2}\hel(\alpha p+\alb q,q)}_{A_4}
\end{align}

 From \cref{eq:d2f}, the terms $A_2$, $A_3$, and $A_4$ are non-positive, whereas $A_1$ is non-negative. We prove that $A_1+A_3\leq0$ and conclude the result. 
\begin{align}\label{ineq:last}
\nonumber
    &A_1 \leq -A_3 \Leftrightarrow
    \frac{(p-q)}{\alpha}\left(\frac{\sqrt{\qb}}{\sqrt{1-\alpha p-\alb q}}-\frac{\sqrt{q}}{\sqrt{\alpha p+\alb q}}\right)
    \leq
    (p-q)^2\left(\frac{1}{\alpha p +\alb q}+\frac{1}{1-\alpha p-\alb q}\right)  
    \Leftrightarrow
    \\&\nonumber
    \left(\frac{p-q}{\alpha}\right)\sqrt{(\alpha p+\alb q)(1-\alpha p-\alb q)}\left(\sqrt{\qb(\alpha p +\alb q)}-\sqrt{q(1-\alpha p-\alb q)}\right)
    \leq
    (p-q)^2
    \Leftrightarrow\\&\nonumber
    (p-q)^2\frac{\sqrt{(\alpha p+\alb q)(1-\alpha p-\alb q)}}{\sqrt{\qb(\alpha p +\alb q)}+\sqrt{q(1-\alpha p-\alb q)}}\leq (p-q)^2 \Leftrightarrow\\&
    \frac{\sqrt{(\alpha p+\alb q)(1-\alpha p-\alb q)}}{\sqrt{\qb(\alpha p +\alb q)}+\sqrt{q(1-\alpha p-\alb q)}}\leq 1
\end{align}
For $\alpha=0$, the inequality reduces to $1/2\leq 1$. For $\alpha=1$, it remains to prove the following inequality for all $p,q\in[0,1]$:
\begin{equation*}
    \frac{\sqrt{p\pb}}{\sqrt{\qb p}+\sqrt{q\pb}}\leq 1
\end{equation*}
This endpoint inequality is equivalent to \cref{ineq:last} for arbitrary $\alpha,p,q$.
Fix $p$. The cases $p=0$ and $p=1$ are immediate, so let $p\in(0,1)$ and define $g(q) = \sqrt{\qb p}+\sqrt{q\pb}$. Then
\begin{align*}
    &g(0) = \sqrt{p}\geq\sqrt{p\pb}\\&
    g(1) = \sqrt{\pb}\geq\sqrt{p\pb}
\end{align*}
The function $g$ is concave, so its endpoint bounds imply $g(q)\geq\sqrt{p\pb}$ for all $q\in[0,1]$.
\end{proof}

\subsubsection{The boundary case $\beta=1/2$}

\begin{lemma}\label{lem:ineq:beta=1/2}
For any $\alpha,p,q\in[0,1]$ with $\alpha\neq1/2$, we have $f(\alpha,1/2,p,q)\geq 0$.
\end{lemma}
\begin{proof}[Proof of \Cref{lem:ineq:beta=1/2}]
    As in the proof of \cref{lem:ineq:beta=alpha}, it is enough to prove the following inequality:
    \begin{equation*}
        \jsa(p,q)\geq\frac{\alpha^2\log(1/\alpha)}{(\alpha-1/2)^2}\hel(\ma,1/2(p+q))
    \end{equation*}
    where $\ma=\alpha p+\alb q$. We prove this by comparing the second derivatives of the functions corresponding to these $f$-divergences. Denoting the function corresponding to $\hel(\ma,1/2(p+q))$ by $f_H(x)$, we have
    \begin{align*}
        f''_H(x) = \frac{\sqrt{2}}{8}\frac{(\alpha-\alb)^2}{(x+1)^{3/2}(\alpha x+\alb)^{3/2}}
    \end{align*}
    Therefore, we get 
    \begin{align*}
        \frac{f''_{JS}(x)}{f''_H(x)}
        =
        \frac{\sqrt{2}\alpha\alb}{(\alpha-1/2)^2}\frac{(x+1)^{3/2}}{x}\sqrt{\alpha x+\alb}
        \geq
        \frac{\sqrt{2}\alpha^2\log(1/\alpha)}{(\alpha-1/2)^2}\frac{(x+1)^{3/2}}{x}\sqrt{\alpha x+\alb}
    \end{align*}
    It is enough to prove that $\frac{(x+1)^{3/2}}{x}\sqrt{\alpha x+\alb}\geq 1$ for all $x\geq0$. The expression diverges as $x\downarrow0$, so it remains to consider $x>0$.
    \begin{align*}
        \frac{(x+1)^{3/2}}{x}\sqrt{\alpha x+\alb}\geq 1
        \Leftrightarrow
        (x+1)^3(\alpha x+\alb)-x^2 \geq 0
    \end{align*}
    Define
    \begin{align*}
        F_\alpha(x)=(x+1)^3(\alpha x+\alb)-x^2.
    \end{align*}
    Direct differentiation gives
    \begin{align*}
        F_\alpha'(x)
        &=
        3(x+1)^2(\alpha x+\alb)
        +\alpha(x+1)^3-2x,\\
        F_\alpha''(x)
        &=
        6(x+1)(2\alpha x+1)-2
        \geq 4.
    \end{align*}
    Moreover, $F_\alpha'(0)=3-2\alpha\geq1$ and $F_\alpha(0)=\alb\geq0$. Hence $F_\alpha'(x)\geq0$ and $F_\alpha(x)\geq0$ for every $x\geq0$, as required.
\end{proof}

\subsubsection{Conclusion of the mixed Hellinger--Jensen--Shannon inequality}

\begin{proof}[Proof of \Cref{thm:mixed-hellinger-js}]
By the reduction above, it remains to prove $f(\alpha,\beta,p,q)\geq0$ for every $p,q\in[0,1]$ and every $\beta\in[0,1/2]$. For fixed $\alpha,p,q$, \Cref{lem:concavity-in-beta} shows that this function is concave in $\beta$ away from the removable point $\beta=\alpha$. The limiting value at that point is non-negative by \Cref{lem:ineq:beta=alpha}. The two endpoint estimates are \Cref{lem:ineq:beta=0} and \Cref{lem:ineq:beta=1/2}. Concavity then gives $f(\alpha,\beta,p,q)\geq0$ throughout the interval $\beta\in[0,1/2]$, and the binary reduction proves the theorem for arbitrary discrete $\cP$ and $\cQ$.
\end{proof}

\subsection{Post-SDPI constant for Bernoulli testing}\label{app:bernoulli-post-sdpi}

\begin{proof}[Proof of \Cref{lem:psdpi_ber}]
    We prove the result for the case where $\frac{\alpha+\beta}{2}\leq1/2$, for the other case the proof is exactly similar, and this is why there is a $\min\{\alpha+\beta,\alb+\betb\}$ in the final result. Based on the \cref{def:post_sdpi} we have
    \begin{align*}
        \eta(P_{Y|X},P_X)
        &=
        \sup_{\substack{X-Y-Z\\ s.t.\ P_{XY}=P_XP_{Y|X}}}\frac{\mi(X;Z)}{\mi(Y;Z)}
        =
        \sup_{\substack{X-Y-Z\\ s.t.\ P_{XY}=P_XP_{Y|X}}}
        \frac{\E_Z\dkl(P_{X|Z}\|\ber(\frac{1}{2}))}{\E_Z\dkl(P_{Y|Z}\|\ber(\frac{\alpha+\beta}{2}))}
    \end{align*}
    Now let $t(z)\triangleq P(Y=1|Z=z)$. Then we get
    \begin{align}
        \sup_{\substack{X-Y-Z\\ s.t.\ P_{XY}=P_XP_{Y|X}}}
        \frac{\E_Z\dkl(P_{X|Z}\|\ber(\frac{1}{2}))}{\E_Z\dkl(P_{Y|Z}\|\ber(\frac{\alpha+\beta}{2}))}
        &\overset{(i)}{=}
        \sup_{\substack{X-Y-Z\\ s.t.\ P_{XY}=P_XP_{Y|X}}}
        \frac{\E_Z\dkl\lp\frac{\alpha}{\alpha+\beta}t(z)+\frac{\alb}{\alb+\betb}\tb(z)\|\frac{1}{2}\rp}{\E_Z\dkl\lp t(z)\|\frac{\alpha+\beta}{2}\rp}
        \nonumber\\&\label{eq:sup_sdpi}
        \overset{(ii)}{=}
        \sup_{t\in[0,1]}\frac{\dkl\lp\frac{\alpha}{\alpha+\beta}t+\frac{\alb}{\alb+\betb}\tb\|\frac{1}{2}\rp}{\dkl\lp t\|\frac{\alpha+\beta}{2}\rp}
    \end{align}
    Where (i) is just rewriting the terms $P(X=1|Z)$ and $P(Y=1|Z)$ in terms of our new notation, and (ii) follows from the fact that $\frac{\sum_i a_i}{\sum_i b_i}\leq\sup_i\frac{a_i}{b_i}$ for $a_i,b_i\geq0$. But here, because we can choose $t(z)$ for every $z$, we get the equality. To compute \cref{eq:sup_sdpi}, we use \cref{fact:dkl_half} to simplify it.
    \begin{align*}
        \sup_{t\in[0,1]}\frac{\dkl\lp\frac{\alpha}{\alpha+\beta}t+\frac{\alb}{\alb+\betb}\tb\|\frac{1}{2}\rp}{\dkl\lp t\|\frac{\alpha+\beta}{2}\rp}
        =
        \sup_{t\in[0,1]}\frac{\dkl\lp\frac{1}{2}+\frac{\alpha-\beta}{(\alpha+\beta)(\alb+\betb)}(t-\frac{\alpha+\beta}{2})\|\frac{1}{2}\rp}{\dkl\lp t\|\frac{\alpha+\beta}{2}\rp}
        \asymp
        s^2\sup_{t\in[-\psi,\psib]}\frac{t^2}{\dkl(\psi+t\|\psi)}
    \end{align*}
    Where we define $s=\frac{\alpha-\beta}{(\alpha+\beta)(\alb+\betb)}$ and $\psi=\frac{\alpha+\beta}{2}$. First, we calculate a lower bound for this term by choosing specific value for $t$. Putting $t=\psib$ we get the lower bound $\frac{\psib^2}{\log(1/\psi)}$. So we have 
    \begin{align*}
        \sup_{t\in[-\psi,\psib]}\frac{t^2}{\dkl(\psi+t\|\psi)}
        \geq
        \frac{1}{4\log(1/\psi)}
    \end{align*}
    
    For the upper bound we consider two cases: 
    \vspace{0.1cm}\\
    \textbf{Case(I): $t\in[-\psi,0]$} From the Taylor expansion we have
    \begin{align}\label{ineq:sdpi_sup_up1}
        \dkl(\psi+t\|\psi)=\sum_{k=2}^\infty\frac{1}{k(k-1)}\lp\frac{(-1)^k}{\psi^{k-1}}+\frac{1}{\psib^{k-1}}\rp t^k
        \geq
        \frac{1}{2\psi\psib}t^2
        \geq
        \frac{1}{2\psi}t^2
        \geq
        \frac{\log(1/\psi)}{4}t^2
    \end{align}
    Where the inequality holds because we assumed $\psi\leq1/2$ and $t\leq0$.
    \vspace{0.1cm}\\
    \textbf{Case(II): $t\in[0,\psib]$}
    We claim that
    \begin{align}\label{ineq:sdpi_sup_up2_sharp}
        \sup_{t\in[0,\psib]}\frac{t^2}{\dkl(\psi+t\|\psi)}
        \leq
        \frac{4}{\log(1/\psi)}.
    \end{align}
    First consider the case $\psi\geq e^{-4}$. For this case we use the Taylor's remainder theorem to get
    \begin{align*}
        \dkl(\psi+t\|\psi)=\frac{1}{2x(1-x)}t^2
        \quad
        \text{for some }
        x\in[\psi,1]
        \Rightarrow
        \dkl(\psi+t\|\psi)
        \geq
        2t^2
        \geq
        \frac{\log(1/\psi)}{4}t^2
    \end{align*}
    So let us assume $\psi\leq e^{-4}$. Again we consider two cases
    \begin{itemize}
        \item $t\in[0,\frac{1}{log(1/\psi)}]$ Here we again use the Taylor's remainder theorem to get
        \begin{align*}
        \dkl(\psi+t\|\psi)=\frac{1}{2x(1-x)}t^2
        \quad
        \text{for some }
        x\in[\psi,\psi+\frac{1}{\log(1/\psi)}]
        \end{align*}
        Now by our choice $\psi+\frac{1}{\log(1/\psi)}\leq1/2$ so the minimum value happens at $x=\psi+\frac{1}{\log(1/\psi)}\leq1/2$ and so we get 
        \begin{align*}
        \dkl(\psi+t\|\psi)
        \geq
        \frac{1}{2(\psi+\frac{1}{\log(1/\psi)})(1-\psi-\frac{1}{\log(1/\psi)})}t^2
        \overset{(i)}{\geq}
        \frac{\log(1/\psi)}{4}t^2
        \end{align*}
        Where (i) follows from the fact that for $\psi\leq e^{-4}$ we have $\psi\log(1/\psi)\leq 1$ and also $1-\psi-\frac{1}{\log(1/\psi)}\leq 1$.
        
        \item $t\in[\frac{1}{\log(1/\psi)},\psib]$ For this case we expand the $\dkl(\psi+t\|\psi)$
        \begin{align}\label{ineq:sdpi_dkl_low}
            \dkl(\psi+t\|\psi)=(\psi+t)\log(\frac{\psi+t}{\psi})-(1-\psi-t)\log(\frac{1-\psi}{1-\psi-t})
            \geq
            t\log(\frac{t}{\psi})-t
        \end{align}
        Where the inequality follows from $1+x\leq e^x$. And as a result we get
        \begin{align*}
            \sup_{t\in[\frac{1}{\log(1/\psi)},\psib]}\frac{t^2}{\dkl(\psi+t\|\psi)}
            \leq
            \sup_{t\in[\frac{1}{\log(1/\psi)},\psib]}\frac{1}{\log(\frac{t}{\psi})-1}
            \leq
            \frac{1}{\log(1/\psi)-\log\log(1/\psi)-1}
            \leq
            \frac{4}{\log(1/\psi)}
        \end{align*}
        Where the first inequality follows from \cref{ineq:sdpi_dkl_low} and the fact that $t\leq1$, and the second inequality follows because $\log\log(1/\psi)+1\leq \frac{3}{4}\log(1/\psi)$ for $\psi\in[0,e^{-4}]$ as we prove here. We need to show that $\log(e\log(x))\leq\frac{3}{4}\log(x)$ for any $x\geq e^4$, which is equal to say $e\log(x)\leq x^{\frac{3}{4}}$. Moving both terms to the right hand side and taking derivative we get $1/x(3/4x^{3/4}-e)$, which is positive for $x\geq e^4$, so we only need to check the inequality for $x=e^4$, which is true as $4e\leq e^3$.
    \end{itemize}
    
    Putting together \cref{ineq:sdpi_sup_up1,ineq:sdpi_sup_up2_sharp} we have the following upper bound
    \begin{align*}
        \sup_{t\in[-\psi,\psib]}\frac{t^2}{\dkl(\psi+t\|\psi)}
        \leq
        \frac{4}{\log(1/\psi)}
    \end{align*}
    And combining the lower and upper bounds that we proved, we get the following
    \begin{align*}
        \frac{1}{2}\frac{(\alpha-\beta)^2}{(\alpha+\beta)^2(\alb+\betb)^2\log(\frac{2}{\alpha+\beta})}
        \leq
        \eta(P_{Y|X},\ber(1/2))
        \leq
        \frac{4\pi^2}{3}
        \frac{(\alpha-\beta)^2}{(\alpha+\beta)^2(\alb+\betb)^2\log(\frac{2}{\alpha+\beta})}
    \end{align*}
\end{proof}
    \begin{fact}\label{fact:dkl_half}
        For any $u\in[-1/2,1/2]$ we have 
        \begin{equation*}
            2u^2\leq\dkl(1/2+u\|1/2)\leq\frac{\pi^2}{3}u^2
        \end{equation*}
    \end{fact}
    \begin{proof}[Proof of \Cref{fact:dkl_half}]
        The proof follows from the Taylor expansion 
        \begin{align*}
            2u^2
            \leq
            \dkl(1/2+u\|1/2)=\sum_{k=1}^\infty\frac{2^{2k-1}}{k(2k-1)}u^{2k}
            \leq
            2u^2\sum_{k=1}^\infty\frac{1}{k^2}=\frac{\pi^2}{3}u^2
        \end{align*}
    \end{proof}

\subsection{Set-disjointness application}\label{app:set-disjointness-application}

    \bravermandisj*

\begin{proof}[Proof of \Cref{lem:brav_disj}]

    Suppose $Y_i\simiid\ber(1-\frac{1}{\sqrt{k}})$ for $i\in[k]$, and note that any protocol $\Pi$ that can compute $\bigwedge_{i\in[k]}b_i$ for any fixed $b=(b_1,\dots,b_k)\in\{0,1\}^k$ with some constant advantage $\delta\in(0,1)$, can also solve hypothesis testing between $\ber(1)$ and $\ber(1-\frac{1}{\sqrt{k}})$ with the probability of error at most $(1-\frac{1}{\sqrt{k}})^k+1/2(1-\delta)$ which is less than $1/2(1-\delta/2)$ for large enough $k$. As a result, we have the following lower bound using \cref{thm:info_comp_ber}

\begin{equation*}
    \mi(Y;\Pi(Y))\geq\frac{\kappa\log(k)}{16}\delta^2
\end{equation*}

Now the goal is to lower bound $\mi(X;\Pi(X)|Z)$ with $\mi(Y;\Pi(Y))$. Note that on average, we would have $\sqrt{k}$ zeros in $X_i$'s, so the intuition is conditioning on one of them shouldn't give that much information. Motivated by this intuition, define the event $A_i\triangleq\{Y_i=0\}$ for $i\in[k]$ and define $S=\sum_{i\in[k]}\mathbbm{1}_{A_i}$ to be the number of zeros in the vector $Y$. Now define the good event to be $G\triangleq\{S\geq\frac{\sqrt{k}}{2}\}$ and note that using Chernoff bound we have $\prob(G^c)\leq \exp(-\frac{\sqrt{k}}{8})$. We can write 

\begin{equation}\label{eq:mi_braverman_proof}
    \mi(Y;\Pi(Y))
    =
    \E_{Y}[\dkl(P_{\Pi|Y}\|P_{\Pi})\mathbbm{1}_{G}]+\E_{Y}[\dkl(P_{\Pi|Y}\|P_{\Pi})\mathbbm{1}_{G^c}]
\end{equation}

We start by showing that the second term is negligible
\begin{align}
    \E_{Y}[\dkl(P_{\Pi|Y}\|P_{\Pi})\mathbbm{1}_{G^c}]&
    =
    \prob(G^c)\E_{Y|G^c}[\dkl(P_{\Pi|Y}\|P_{\Pi})]
    \nonumber\\&=
    \prob(G^c)\mi(Y;\Pi(Y)|G^c)+\prob(G^c)\dkl(P_{\Pi|G^c}\|P_\Pi)
    \nonumber\\&\leq
    \prob(G^c)\mi(Y;\Pi(Y)|G^c)+\mi(\mathbbm{1}_G;\Pi(Y))
    \leq
    k e^{-\frac{\sqrt{k}}{8}}+\log(2)
    \leq
    2\log(2)\label{ineq:mi_braverman_i_upperbound}
\end{align}
Where the last inequality holds for large enough $k$. Now we upper-bound the first term in \cref{eq:mi_braverman_proof}. Note that on the event $G$ we have $1\leq\frac{2}{\sqrt{k}}S$
\begin{align}
    \E_{Y}[\dkl(P_{\Pi|Y}\|P_{\Pi})\mathbbm{1}_{G}]
    &\overset{(i)}{\leq}
    \frac{2}{\sqrt{k}}\sum_{i\in[k]}\E_{Y}[\dkl(P_{\Pi|Y}\|P_{\Pi})\mathbbm{1}_{A_i}]
    =
    \frac{2}{k}\sum_{i\in[k]}\E_{Y|A_i}\dkl(P_{\Pi|Y}\|P_{\Pi})
    \nonumber\\&=
    \underbrace{\frac{2}{k}\sum_{i\in[k]}\mi(Y;\Pi(Y)|A_i)}_{=2\mi(X;\Pi(X)|Z)} + \frac{2}{k}\sum_{i\in[k]}\dkl(P_{\Pi|A_i}\|P_{\Pi})
    \nonumber\\&\overset{(ii)}{\leq}
    2\mi(X;\Pi(X)|Z)+\frac{2}{\sqrt{k}}\sum_{i\in[k]}\mi(\mathbbm{1}_{A_i};\Pi(Y))
    \overset{(iii)}{\leq}
    2\lp\mi(X;\Pi(X)|Z)+\frac{1}{\sqrt{k}}\mi(Y;\Pi(Y))\rp\label{ineq:mi_braverman_ii_upperbound}
\end{align}
Where (i) follows from substituting $\mathbbm{1}_{G}$ with $\frac{2}{\sqrt{k}}S$, (ii) follows from the fact that condition on $A_i$, Y has the same distribution as $X$ condition on $Z=i$, and also the fact that $\prob(A_i)=\frac{1}{\sqrt{k}}$ so we can upper bound the second sum with the sum of mutual information by multiplying and dividing by $\sqrt{k}$, and (iii) follows from the fact that $\mathbbm{1}_{A_i}=1-Y_i$ and the fact that $Y_i$'s are independent so we have $\mi(Y;\Pi(Y))\geq\sum_{i\in[k]}\mi(Y_i;\Pi(Y))$.
So putting \cref{ineq:mi_braverman_i_upperbound,ineq:mi_braverman_ii_upperbound} together, for large enough $k$ we have 
\begin{equation*}
    \mi(Y;\Pi(Y))
    \leq
    2\lp\mi(X;\Pi(X)|Z)+\frac{1}{\sqrt{k}}\mi(Y;\Pi(Y))+\log(2)\rp
    \Rightarrow
    \mi(X;\Pi(X)|Z)
    \geq
    \frac{\kappa\delta^2}{128}\log(k)
\end{equation*}
\end{proof}

\end{document}